\documentclass{AIMS}

\usepackage{latexsym,amsmath,amssymb,graphics,setspace}
\usepackage{enumerate,framed,sidecap}
\usepackage[colorlinks]{hyperref}
\usepackage{url}

\usepackage[all]{xy}
\usepackage[pdftex]{graphicx}

\newcommand{\beq}{\begin{eqnarray*}}
\newcommand{\eeq}{\end{eqnarray*}}

\newcommand{\Diff}{\operatorname{Dif{}f}}
\newcommand{\Emb}{\operatorname{Emb}}

\newcommand{\rem}[1]{}

\newcommand{\mg}{{\mathfrak{g}}}

\newcommand{\mX}{{\mathfrak{X}}}

\newcommand{\cA}{{\mathcal{A}}}
\newcommand{\cB}{{\mathcal{B}}}

\newcommand{\cG}{{\mathcal{G}}}

\newcommand{\cI}{{\mathcal{I}}}

\newcommand{\cM}{{\mathcal{M}}}
\newcommand{\cN}{{\mathcal{N}}}

\newcommand{\cS}{{\mathcal{S}}}

\newcommand{\cU}{{\mathcal{U}}}

\newcommand{\mI}{{\mathbb{I}}}

\newcommand{\mR}{{\mathbb{R}}}

\newcommand{\al}{\alpha}

\newcommand{\de}{\delta}

\newcommand{\sig}{\sigma}

\newcommand{\om}{\omega}
\newcommand{\Om}{\Omega}

\newcommand{\dd}[2]{\frac{d #1}{d #2}}

\newcommand{\prt}{\partial}
\newcommand{\dede}[2]{\frac{\de #1}{\de #2}}

\newcommand{\eval}[2]{{\left. #1 \right|_{#2}}}

\newcommand{\Ad}{{\mbox{Ad}}}
\newcommand{\ad}{{\mbox{ad}}}

\newcommand{\id}{{\mbox{id}}}

\newcommand{\hor}{{\mbox{\rm Hor}}}

\newcommand{\lsb}{\left[}
\newcommand{\rsb}{\right]}
\newcommand{\llsb}{\left[ \! \left[}
\newcommand{\rrsb}{\right] \! \right]}
\newcommand{\bllsb}{\left[ \!\! \left[}
\newcommand{\brrsb}{\right] \!\! \right]}

\newcommand{\lp}{\left(}
\newcommand{\rp}{\right)}

\newcommand{\scp}[2]{{\left\langle {#1}\, , \, {#2}\right\rangle}}

\newtheorem{thm}{Theorem}[section]

\newtheorem{prop}[thm]{Proposition}
\newtheorem{lem}[thm]{Lemma}
\newtheorem{rema}[thm]{Remark}
\newtheorem{problem}[thm]{Problem}
\newtheorem{example}[thm]{Example}

\title{Un-reduction}
\author[M. Bruveris, D.C.P.\ Ellis, F. Gay-Balmaz and D.D.\ Holm]{}

\subjclass{37K05, 37K65}
\keywords{Lagrangian reduction, force fields, un-reduction, image matching}
\email{m.bruveris08@imperial.ac.uk}
\email{david.ellis1@imperial.ac.uk}
\email{gaybalma@lmd.ens.fr}
\email{d.holm@imperial.ac.uk}

\begin{document}

\maketitle

\centerline{\scshape Martins Bruveris, David C.P. Ellis and Darryl D. Holm }
\medskip
{\footnotesize
 	\centerline{Department of Mathematics}
	\centerline{Imperial College London}
	\centerline{London, SW7 2AZ UK}
}

\medskip

\centerline{\scshape Fran\c{c}ois Gay-Balmaz}
\medskip
{\footnotesize
 \centerline{Laboratoire de M\'et\'eorologie Dynamique}
 \centerline{Ecole Normale Sup\'erieure de Paris}
 \centerline{75005 Paris, France.}
}

\bigskip

\begin{abstract}
\noindent This paper provides a full geometric development of a new technique called {\em un-reduction}, for dealing with dynamics and optimal control problems posed on spaces that are unwieldy for numerical implementation.  The technique, which was originally concieved for an application to image dynamics, uses Lagrangian reduction by symmetry in reverse.  A deeper understanding of un-reduction leads to new developments in image matching which serve to illustrate the mathematical power of the technique.
\end{abstract}


\section{Introduction}\label{sec:Intro}

Recently there has been considerable interest in geodesic shape matching, as discussed in \cite{BaHaMi2010,CoHo2009,MiMu07} and references therein.  This interest has grown from the desire to assist applications in medical imaging by providing a coherent, quantitative method for comparing shapes. See \cite{Yo2010,YoArMi2009} for extensive discussions of the background and mathematical basis of this endeavor, which is known as \emph{computational anatomy}. 

For our purposes, a \emph{shape} is taken to be an embedded or immersed submanifold of an ambient space. Given two smooth manifolds $\cM$ and $\cN$, the shape space of $\cM$-type submanifolds of $\cN$ may be identified with the quotient space $\Sigma:=\Emb\lp \cM,\cN\rp/\Diff\lp \cM\rp$.  That is, the space of embeddings of $\cM$ into $\cN$ up to re-parametrization of $\cM$.
An example of particular interest is the shape space of all closed, simple, planar unparametrized curves, given by $ \Sigma = \Emb\lp S ^1 ,\mathbb{R}  ^2 \rp/\Diff\lp S ^1\rp$

In order to compare two shapes, a natural way to proceed is to construct a path between them.  This is known as the {\em matching problem}.  For example, the matching problem for closed, simple planar curves may be stated as follows:  Let $\rho_0$ and $\rho_1$ be two smooth submanifolds in $\mR^2$ of $S^1$-type.  Find a path of submanifolds of $S^1$-type, $\rho(t)$, such that $\rho(0) = \rho_0$ and $\rho(1) = \rho_1$.  The preferred path is selected in such a way that it minimizes a certain functional, as illustrated in the below.
\vspace{1em}

\begin{problem}[The matching problem]\label{IntroProblem} {\rm Minimize the functional
\[
\int_0^1 \ell^\Sigma\lp \rho, \dot \rho \rp\, \mbox{d}t
\]
subject to the boundary conditions $\rho(0) = \rho_0$ and $\rho(1) = \rho_1$, where $ \Sigma $ denotes the space of all submanifolds of $S^1$-type in the plane and $\ell^ \Sigma :T \Sigma \rightarrow \mathbb{R}  $ is a Lagrangian defined on the tangent bundle of $ \Sigma $.}
\end{problem}
\vspace{1em}

As is well known, the solution to the matching problem satisfies the Euler-Lagrange equations
\begin{equation}
\frac{d}{dt}  \dede{\ell^\Sigma}{\dot \rho} - \dede{\ell^\Sigma}{\rho} =0.
\label{eq:EL}
\end{equation}

In order for these equations to make sense rigorously, $\Sigma$ must be endowed with a smooth manifold structure.  This is achieved in \cite{KrMi1997,Mi1980}.  The key concern with the matching problem, however, is that numerical implementation of \eqref{eq:EL} is rendered unfeasible by a lack of coordinates for the shape space $ \Sigma $. 
The task for this paper is to investigate alternative formulations of the matching problem in such a way that their solutions are tractable to practical implementation.

\medskip

\paragraph{\bf Introducing un-reduction} Recent progress has been made for the matching problem with simple, planar, closed curves in the geodesic case, with $\ell^\Sigma(\rho,\dot \rho) = \cG^\Sigma_\rho\lp \dot \rho,\dot \rho\rp$ where $\cG^\Sigma$ is a Riemannian metric on the $\Sigma=\Emb\lp S^1, \mR^2\rp/\Diff\lp S^1\rp$, called shape space. This approach involved formulating a related geodesic problem on $\Emb\lp S^1, \mR^2\rp$ rather than $\Sigma$, \cite{CoHo2009,MiMu07}.  These papers study the geodesic problem by adapting techniques commonly used for reduction by symmetry of variational principles, as in \cite{Mo1993}, in a novel way.  In general, given a principal $G$-bundle $\pi:Q \to \Sigma$ and a $G$-invariant Lagrangian $L:TQ \to \mR$, one may write down a reduced variational principle on $TQ/G$, as in \cite{CeMaRa2001}.  This procedure is called \emph{Lagrangian reduction by symmetry}.  The novel insight used in \cite{CoHo2009,MiMu07} is to use this procedure in reverse.  That is, one is interested in geodesics on the base space $ \Sigma $, but instead of applying a variational principle directly, one formulates a geodesic problem on $Q$ that reduces to the desired geodesic problem on $\Sigma$.  That is, a form of \emph{un-reduction} is employed, and the resulting equations on $Q$ are more convenient to deal with numerically. This approach has also been applied recently in \cite{BaHaMi2010} to achieve promising results for higher dimensional geodesic problems.

The un-reduction procedure is distinct from the inverse process to Lagrangian reduction by symmetry, called {\em reconstruction}.  Lagrangian reduction by symmetry relates the Euler-Lagrange (EL) equations on $TQ$ with a new set of equations on $TQ/G$, known as the Lagrange-Poincar\'e (LP) equations.  Correspondingly, the reconstruction procedure relates solutions of the LP equations on $TQ/G$ with solutions of a set of EL equations on $TQ$.  The goal of un-reduction, however, is to find a parametric family of equations, called {\em un-reduction equations}, on $TQ$ whose solution projects onto those of a set of EL equations on $T\lp Q/G\rp$.  Therefore, un-reduction and reconstruction only coincide when the component of the LP equations on $T\lp Q/G\rp$ coincides with the desired set of EL equations.  In fact, even when this coincidence occurs, reconstruction turns out to be a sub-manifold of an infinite dimensional space of un-reduction equations.  Nevertheless, un-reduction equations that coincide with reconstruction do arise in the case of geodesic equations with horizontality conditions enforced.  This fact has enabled the development of horizontal shooting methods for image dynamics in \cite{CoHo2009,MiMu07}.  The term `un-reduction', which was coined in \cite{CoHo2009}, reflects the fact that the technique lifts EL equations from $T\lp Q/G\rp$ to $TQ$, this mirrors the way reduction projects EL equations from $TQ$ to $TQ/G$.

\begin{SCfigure}%
\centering
\caption{A trajectory in $Q$ from $q_0$ to $q_1$ that projects via $\pi$ to a trajectory in $\Sigma$ from $\rho_0$ to $\rho_1$.  Un-reduction poses a boundary value problem in $Q$ such that the projected trajectories are solutions to the Euler-Lagrange equations in $\Sigma$.}%
\includegraphics[width=0.43\columnwidth]{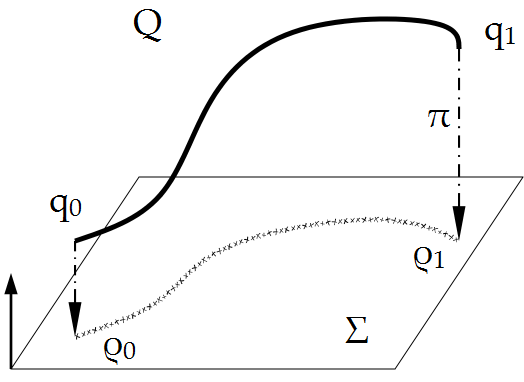}%
\label{fig:principalBundle}%
\end{SCfigure}
Some numerical difficulties persist, in that, although the geodesic in shape space may be successfully represented, its parametrization may still evolve in an undesirable way.  For example, when the geodesic equations are implemented on $\Emb\lp S^1,\mR^2\rp$, phenomena such as clustering of data points and bad parametrizations present themselves, as shown in Figure $1$ of \cite{CoHo2009}.  One approach developed in \cite{CoHo2009} and \cite{ClCoPe2010} deals with these issues by adding a step in the numerical procedure that reparametrizes the initial conditions to evolve onto a prescribed parametrization of the target shape.  This initial reparametrization procedure commutes with evolution along the matching trajectory, and may also iterated a number of times throughout the evolution without significantly changing the problem.  This iterative approach effectively breaks the matching problem up into a number of sub-problems.  While this method works well in practice, it is {\em extrinsic} in nature, and sidesteps the problem with the matching equations rather than dealing with it directly.  It would be preferable to seek matching equations for whom the parametrization problem never arose.  That is, for the parametrization issues to be dealt with {\em intrinsically}.  Indeed, an intrinsic method is also described in (loc. cit.), which involves relaxing constraints on the initial conditions.  The extra initial degrees of freedom so gained allow one to match any boundary data irrespective of parameterisation, however nothing is said about the parametrization dynamics along the trajectory.  One of our aims in developing un-reduction is to provide the necessary tools to derive matching equations that intrinsically control parametrization dynamics along the entire trajectory.  Such a derivation is given in \S\ref{PlaneCurves}, which represents, as far as we are aware, the first geodesic matching equations with dynamic, intrinsic parametrization control.

The papers reporting progress so far, such as \cite{BaHaMi2010,CoHo2009,MiMu07}, have constrained themselves to geodesic problems.  More precisely, geodesics on $ \Sigma $ are obtained by projection of horizontal geodesics on $Q$, relative to a $G$-invariant Riemanian metric on $Q$ that projects to the given metric on $ \Sigma $. Whilst the geodesic problem is important, it may be useful to match curves in a variety of other ways depending on the application.  We shall find that in order to achieve the desired goals in parametrization dynamics, the geodesic properties must be partially sacrificed.  Further, in certain situations it may be useful to assert that the curve dynamics respect different properties such as that of being described as a graph in the ambient space.  For future work we are thinking, for example, of the geometric splines approach for image analysis as recently discussed in \cite{GBHoMeRaVi2010,Vi2009}, which also fits naturally into the un-reduction framework that we develop here. However, applications to image analysis using geometric splines is beyond our present scope and will be discussed elsewhere.  The framework used in developing geodesic matching needs to be extended in order to incorporate non-geodesic properties.  Our motivation in this paper is to provide one possible extension by developing extra geometric tools, and to demonstrate their use on an established problem from curve matching.

\medskip

\paragraph{\bf A simplified problem.} To illustrate the idea of the paper in a simple case, where the technicalities of adjoint bundles and connections are not needed, let us consider the following problem. We start with a manifold $\mathcal{Q} $ on which a Lie group $\mathcal{G} $ acts and we denote the quotient $\mathcal{M} =\mathcal{Q} /\mathcal{G} $ and the projection $\pi:\mathcal{Q} \to \mathcal{M} $. The problem we consider is to find the critical points of a smooth function $f$ on $\mathcal{M} $. The idea of un-reduction is to consider instead a function $\hat f$ on $\mathcal{Q} $, which is connected to $f$ via $\hat f = f \circ\pi + g$ with $g$ some function on $\mathcal{Q} $. By choosing a connection, we can split the tangent bundle as $T\mathcal{Q}  = \ker T\pi \oplus H$, with $H$ the horizontal bundle. Critical points of $\hat f$ satisfy the equations
\begin{align*}
d(f\circ\pi)|_H &= -dg|_H & dg|_{\ker T\pi} = 0.
\end{align*}
Since we were originally interested in the problem $d f = 0$, we introduce a force $F$ and solve the problem $d\hat f = F$ on the total space, where $F$ is chosen in such a way that it cancels the $dg$ term. This in a nutshell is the idea of un-reduction. We will apply these ideas to the case where $ \mathcal{Q} $ is a set of curves in a manifold $Q$ on which a Lie group $G$ acts freely and properly, and $f$ is the action associated to some Lagrangian $\ell:T \Sigma \rightarrow \mathbb{R} $, where $\Sigma= Q/G$, so that $df=0$ corresponds to the Euler-Lagrange equations for $\ell$.

\medskip

\paragraph{\bf The purpose of this paper} 
The objective of this paper is to investigate the extent to which the un-reduction approach may be used in the design of numerical methods, by investigating its geometrical context.  Our investigation, inspired by the novel approach taken in \cite{CoHo2009} that introduced reduction-by-symmetry techniques in reverse for the particular case of geodesic simple, closed, planar curve matching, reveals a rich geometric framework for un-reduction, which turns out to be broader in scope than the pioneering curve matching examples.  The essential features of an un-reduction procedure are highlighted in purely geometric terms, giving clarity and rigorously capturing the generic notion of `doing reduction-by symmetry in reverse'.  When cast in the language of geometry, the un-reduction technique is sufficiently general that one may expect it to find many other productive uses, even in different fields such as data assimilation.  See also \cite{CoHo2009} for additional outlook toward further potential applications of un-reduction.

The general theory re-applied to the curve matching problem in \S\ref{PlaneCurves} goes some way towards answering the call for greater control over parametrization dynamics in \cite{CoHo2009}.  For example, our investigation of the geometry shows that the problem of quality control in parametrizations, and the extension to include potential forces may be addressed \emph{simultaneously}.  Further, the general un-reduction algorithms admit the introduction of a family of exogenous design factors that may be used either to formulate optimal control problems based on parametrization dynamics, or for modeling purposes.  This extra modeling capability may then lead to greater functionality of the solution, such as adaptive matching algorithms in which data points may evolve to preserve properties such as uniform parametrization over time.

After introducing the required geometrical background in \S\ref{sec:Geometry} and explaining the basic ideas in the method of reduction by symmetry in \S\ref{sec:Reduction}, we look at the approaches for applying the symmetry-reduction method \emph{in reverse} as un-reduction in \S\ref{sec:Dist}.  Having achieved a good understanding of the fundamental issues, a general procedure for applying the un-reduction method is then developed in \S\ref{sec:strategies}.   Parallels are drawn both with particular cases of un-reduction and with reduction by symmetry.  Finally, in \S\ref{PlaneCurves} a new un-reduction procedure is applied to the curve matching problem, and comparisons are made with previous treatments.

\rem{
Let's have a section called ``Main themes and results of the paper", OK? \\[2mm]
In this new themes-and-results section to be inserted here, let's lay out the main themes of geometrical thinking that we apply in this paper, e.g., invariants of curves, momentum maps, the importance of keeping track of horizontal and vertical motions, since un-reduction moves in the vertical direction -- and the key theorems to which they lead. For example, I thought the interpretation of their zero-level sets as horizontality conditions in Section \ref{sec:geo} was particularly pleasing. \\[2mm]
We write about the geometrical approach to numerical algorithms because we believe it prepares the way for computing to provide insight, not just crunch numbers.
}

\section{Review of geometric constructions}\label{sec:Geometry}

This section reviews the necessary geometric tools for the formulation of un-reduction.  For a more in depth overview and application of these tools to Lagrangian reduction, see \cite{CeMaRa2001}.

Consider a free and proper right action
\[
 Q \times G \rightarrow Q,\quad (q,g)\mapsto q \cdot g
\]
of a Lie group $G$ on a manifold $Q$, and denote by $\pi: Q\rightarrow \Sigma := Q/G$ the associated principal $G$-bundle.
Let $\mg $ be the Lie algebra of the Lie group $G$.  A {\em connection form} $\cA$ on $\pi$ is a $\mg$-valued one-form that satisfies the following properties:
\begin{enumerate}
        \item $\cA( q \cdot \xi ) = \xi,\;$ for all $\xi \in \mathfrak{g} $ and $ q\in Q$\label{prop:con1},
        \item $\cA\lp v_q\cdot g\rp = \Ad_{g^{-1}}\cA(v_q),\;$ for all  $v_q \in T_qQ$ and $g\in G$\label{prop:con2},
\end{enumerate}
where
\begin{align} \label{infinitesimal_generator}
q \cdot \xi:=\xi_Q(q):= \left.\frac{d}{dt}\right|_{t=0}q \cdot \exp(t\xi)
\end{align} 
is the infinitesimal generator associated to the Lie algebra element $ \xi \in \mathfrak{g}$ and $v_q \cdot g$ denotes the tangent lifted action of $G$ on the tangent bundle $TQ$.

A connection form splits the tangent bundle into vertical and horizontal distributions.  That is, $TQ = VQ\oplus HQ$ where $VQ$ and $HQ$ are the distributions in $TQ$ defined fiberwise by
\[
V_qQ = \ker T_q\pi, \qquad H_qQ = \ker \cA_q.
\]

Since $H_qQ$ is complementary to $\ker T_q\pi$, it constitutes an Ehresmann connection on $Q$, and property \ref{prop:con2} ensures the connection $HQ$ is $G$-equivariant in the sense that $H_{q\cdot g}Q =  H_qQ \cdot g$.  Such a connection is called a {\em principal connection}.  Connection forms are in one-one relationship with principal connections. We will denote by
\[
P^v:TQ\rightarrow VQ \quad\text{and}\quad P^h:TQ\rightarrow HQ
\]
the projections associated to the decomposition $TQ = VQ\oplus HQ$.

The {\em horizontal lift operator}, $\hor_q:T_{\rho}\Sigma \to H_qQ$, $\rho = \pi\lp q\rp$, is by definition the inverse of the isomorphism $\eval{T_q\pi}{H_qQ}: H_qQ \to T_\rho \Sigma$. 
The horizontal lift operator is $G$-equivariant in the sense that $\hor_{q\cdot g}\lp v_\rho\rp = \lp\hor_{q}\lp v_\rho\rp\rp\cdot g$.
Integrability of $HQ$ is measured by the curvature form of $\cA$, which is defined by
\[
\cB\lp v_q, u_q\rp = \mathbf{d} \cA\lp v_q, u_q\rp + \lsb \cA\lp v_q\rp, \cA\lp u_q\rp\rsb.
\]
The principal connection $HQ$ is integrable if and only if $\cB = 0$. For more discussion of these points, see \cite{CeMaRa2001} and references therein.  Note that $\cB$ is a \emph{horizontal} form, that is $\cB\lp u_q, v_q\rp = 0$ if either $u_q \in V_qQ$ or $v_q \in V_qQ$.

\subsection{Associated bundles}

An {\em associated vector bundle} is a quotient manifold of the form $Q\times_G V := \lp Q\times V\rp/G$ where $V$ is a vector space upon which $G$ acts linearly, and $G$ acts on $Q\times V$ by diagonal action.  The equivalence class of $\lp q, v\rp \in Q\times V$ is denoted $\llsb q, v\rrsb_V$.  A vector bundle structure $\tau_V:Q\times_G V \to \Sigma$ is given by
\[
\tau_V\lp \llsb q, v\rrsb_V\rp = \pi\lp q\rp.
\]
The vector space operations are defined fibrewise by
\[
\llsb q, v\rrsb_{V} + \lambda\llsb q, u\rrsb_V = \llsb q, v + \lambda u\rrsb_V,\quad\text{for all $ \lambda \in \mathbb{R}  $}.
\]
Note that $q$ must be the same for each element representative in this relation.

A connection form $\cA$ on $Q$ introduces a covariant derivative of curves in $Q\times_G V$ which reads
\begin{equation} \label{general_formula_covder}
D_t \llsb q\lp t\rp, v \lp t\rp \rrsb_V = \llsb q(t) , \dot v(t) - v(t)\cdot \cA\lp q(t), \dot q(t) \rp \rrsb_V
\end{equation} 
where $\cdot$ denotes the infinitesimal action of $\mathfrak{g} $ on $V$.

The {\em adjoint bundle} is the associated vector bundle with $V=\mg$ under the adjoint action by the inverse, $\xi \mapsto \Ad_{g^{-1}}\xi$, and is denoted $\Ad \, Q := Q \times_G\mg$.  The adjoint bundle is special, since it is a Lie algebra bundle. That is, each fibre is a Lie algebra with the Lie bracket defined by
\[
\lsb \llsb q, \xi \rrsb_\mg, \llsb q, \eta\rrsb_\mg\rsb = \llsb q, \lsb \xi, \eta\rsb \rrsb_\mg.
\]

Sections of the adjoint bundle correspond to $G$-invariant vertical vector fields on $Q$:  Consider a section $\bar \xi\lp \rho\rp = \llsb q, \xi(q) \rrsb_\mg$ where $\pi\lp q\rp = \rho$ and $\xi:Q \to \mg$ is a smooth, $G$-equivariant map.  The section $\bar\xi$ generates a vertical, $G$-invariant vector field $\bar\xi_Q \in \mX\lp Q\rp$ according to
\[
\bar\xi_Q\lp q\rp :=   q \cdot \lp\xi(q)\rp=\xi(q)_Q(q),
\]
where $\cdot$ denotes the infinitesimal action of $ \mathfrak{g}  $ on $Q$, see \eqref{infinitesimal_generator}.
The properties of verticality and $G$-invariance are easily verified.

The {\em coadjoint bundle} is the associated vector bundle with $V= \mg^*$ under the coadjoint action, $\mu \mapsto \Ad^*_{g}\mu$, and is denoted $\operatorname{Ad}^*Q = Q \times_G \mg^*$.  There is a natural pairing between the adjoint and coadjoint bundles given by
\[
\scp{\llsb q, \mu\rrsb_{\mg^*}}{\llsb q, \xi \rrsb_{\mg}} = \scp{\mu}{\xi}_{\mg^*\times \mg}.
\]

By using the general formula \eqref{general_formula_covder}, one obtains the following expression for 
the covariant derivatives of curves in $\operatorname{Ad}Q$ and $\operatorname{Ad}Q^*$:
\begin{align*}
D_t\llsb q\lp t\rp, \xi\lp t\rp \rrsb_{\mg} &= \bllsb q\lp t\rp, \dot \xi\lp t\rp + \lsb \cA\lp q\lp t\rp, \dot q\lp t\rp \rp, \xi\lp t\rp \rsb \brrsb_{\mg}\\
D_t\llsb q\lp t\rp, \mu\lp t\rp \rrsb_{\mg^*} &= \bllsb q\lp t\rp, \dot \mu\lp t\rp - \ad^*_{\cA\lp q\lp t\rp, \dot q\lp t\rp \rp} \mu\lp t\rp \brrsb_{\mg^*}.
\end{align*}
Therefore, for any curve $\bar \mu\lp t\rp$ in $\operatorname{Ad}^*Q$ and $\bar \xi \lp t\rp$ in $\operatorname{Ad}Q$ covering the same curve $ \rho (t)\in \Sigma$, we have
\[
\dd{}{t}\scp{\bar\mu(t)}{\bar \xi(t)} = \scp{D_t \bar \mu\lp t\rp}{\bar \xi(t)} + \scp{ \bar \mu\lp t\rp}{D_t \bar \xi(t)}.
\]
The connection form $\cA$ on $Q$ induces a map $\bar \cA: TQ \to \operatorname{Ad}Q$ defined by
\[
\bar \cA\lp v_q\rp = \llsb q, \cA\lp v_q\rp\rrsb_{\mg}.
\]
Note, in particular, that
\[
\bar\cA \circ \bar\xi_Q = \bar\xi
\]
for all sections $\bar \xi \in \Gamma \lp \operatorname{Ad}Q\rp$.

Since the curvature form $ \mathcal{B} $ is horizontal, it induces an $ \operatorname{Ad}Q$-valued two-form $\bar {\mathcal{B}} $ on $ \Sigma $ defined by
\[
\bar{ \mathcal{B} }_\rho(u_ \rho , v_ \rho )=\bllsb q, \mathcal{B} _q( u _q , v _q ) \brrsb_{\mg},\quad u_ \rho , v_ \rho \in T_ \rho \Sigma ,
\]
where $u _q , v_q $ are arbitrary vectors in $T_qQ$ such that $T_q(u_q)=u_ \rho $ and $T_q \pi ( v _q )= v _\rho $. The two-form $\bar{ \mathcal{B} }$ is called the \emph{reduced curvature form}.

Finally, recall that the {\em cotangent lift momentum map}, $ \mathbf{J} :T^*Q \to \mg^*$, is defined by
\[
\scp{\mathbf{J} \lp \alpha_q \rp}{\xi}_{\mg^*\times \mg} = \scp{\alpha_q}{q\cdot \xi}_{T^*Q\times TQ}
\]
for all $\xi \in \mg$ and $\alpha_q \in T_q^*Q$. The cotangent lift momentum map induces a map $\bar { \mathbf{J} }: T^*Q \to \operatorname{Ad}^*Q$ defined by
\[
\bar{ \mathbf{J} }\lp \alpha_q\rp = \llsb q, \mathbf{J} \lp \alpha_q\rp \rrsb_{\mg^*}.
\]

\section{Lagrangian reduction with a force field}\label{sec:Reduction}

Un-reduction is closely related to classical Lagrangian reduction by symmetry.  This section briefly reviews the main results of Lagrangian reduction by symmetry.  The description given here by no means does the topic justice and the interested reader is referred to \cite{CeMaRa2001,CeHoMaRa1998,ElGBHoRa2009,HoMaRa1998,MaRa2002} to name just a few works on the subject. We slightly generalize the Lagrange reduction theorem by including the effect of an equivariant force field. Such a field is encoded by a fiber-preserving map $F : TQ\rightarrow T^*Q$ over the identity, see \S7.8 in \cite{MaRa2002}. The $G$-equivariance property reads
\[
F(v_q \cdot g)= F(v _q ) \cdot g,\quad\text{for all $g\in G$},
\]
where $\cdot $ means the tangent and cotangent lifted actions, respectively.

Let $L:TQ \to \mR$ be a $G$-invariant Lagrangian under the tangent lifted action of $G$ on $TQ$ and consider the associated reduced Lagrangian $\ell:TQ/G \to \mR$. Fix a connection $ \mathcal{A} $ on $Q$ and consider the vector bundle isomorphism $\al_\cA: TQ/G \to T\Sigma \oplus \Ad \, Q$ over $ \Sigma $ given by
\begin{equation}
\al_ \mathcal{A} \lp \lsb v_q\rsb_G\rp = T\pi\lp v_q\rp \oplus \bar \cA\lp v_q\rp.
\label{eq:bundleiso}
\end{equation}
The inverse is
\[
\al^{-1}_ \mathcal{A} \lp u_\rho \oplus \bar \xi\rp = \lsb \hor_q u_\rho + \bar\xi_Q\lp q\rp\rsb_G.
\]
Thus the reduced Lagrangian $\ell$ may be regarded as a map $\ell:T\Sigma\oplus \operatorname{Ad}Q \to \mR$. 

Consider a $G$- equivariant force field $F:TQ\rightarrow T^*Q$. We define the reduced force fields $F^{\Sigma}$ and $F^{\operatorname{Ad}}$ by
\begin{align}
F^{\operatorname{ Ad}} &: T \Sigma \oplus \operatorname{ Ad}Q\rightarrow \operatorname{ Ad}^*Q,\quad F^ \Sigma ( \rho , \dot \rho , \bar{\sigma}):=\bar{ \mathbf{J} }( F (q, \dot q))\\
F^ \Sigma &: T \Sigma \oplus \operatorname{ Ad}Q\rightarrow T^* \Sigma ,\quad \left\langle F^ \Sigma ( \rho , \dot \rho , \bar{\sigma}), v_ \rho \right\rangle :=\left\langle  F( q, \dot q), \operatorname{Hor}_q(v_ \rho )\right\rangle, 
\end{align}
where $(q, \dot q)\in TQ$ are such that $ \alpha _{\mathcal{A}}\left(  [q, \dot q]_G \right) =( \rho , \dot \rho , \bar{\sigma})$. Note that by $G$-equivariance of $F$, the right hand side does not depend on the choice of $(q, \dot q)$ in the equivalence class. The force field $F$ is related to  $F^{ \Sigma }$ and $F^{\operatorname{Ad}}$as follows:
\begin{align*}
\left\langle F(q,\dot q),w_q \right\rangle &=\left\langle F(q, \dot q ),  \operatorname{Hor}_q(T\pi(w_q))\right\rangle  + \left\langle F(q, \dot q),  \left( \mathcal{A} (w_q) \right) _Q(q) \right\rangle \\
&=\left\langle  \operatorname{Hor}^*_q( F(q, \dot q)), T\pi(w_q) \right\rangle +\left\langle\mathbf{J} (F(q, \dot q)), \mathcal{A} (w_q) \right\rangle_{ \mathfrak{g}  ^* \times \mathfrak{g}  }\\
&= \left\langle F^\Sigma (\rho, \dot \rho , \bar{ \sigma }), T\pi(w_q)\right\rangle + \left\langle \bar {\mathbf{J}}(F(q, \dot q)), \bar{ \mathcal{A} }(w_q)\right\rangle \\
&= \left\langle F^\Sigma (\rho, \dot \rho , \bar{ \sigma }), T\pi(w_q)\right\rangle+ \left\langle F^{\operatorname{Ad}}(\rho, \dot \rho , \bar{ \sigma }), \bar{ \mathcal{A} }(w_q) \right\rangle. 
\end{align*}
It will be convenient to define also the vertical and horizontal part of $F$ by considering the decomposition $T^*Q=VQ^\circ\oplus HQ^\circ$, where
\begin{align} 
VQ^\circ:&=\{ \alpha _q\in T^*Q\mid \left\langle \alpha _q , v _q \right\rangle =0,\;\text{for all $v_q\in V_qQ$}\}\quad\text{and}\\
HQ^\circ:&=\{ \alpha _q\in T^*Q\mid \left\langle \alpha _q , v _q \right\rangle =0,\;\text{for all $v_q\in H_qQ$}\},
\end{align} 
are the annihilators of the vertical and horizontal distributions, respectively. We thus define
\[
F^h:TQ\rightarrow VQ^\circ\quad\text{and}\quad F^v:TQ\rightarrow HQ^\circ
\]
as $F^h(v_q):=P_h^*(F(v_q))\in VQ^\circ$ and $F^v(v_q):= P_v^*(F(v_q))\in HQ^\circ$ and obtain the relations
\begin{align}\label{relation_Fh_Fv}
F^h(v_q)=\pi ^\ast \left( F^ \Sigma ( \alpha _ \mathcal{A} (v_q))\right) \quad\text{and}\quad F^v(v_q)=\bar{ \mathcal{A} }^*\left( F^{ \operatorname{ Ad}}( \alpha _ \mathcal{A} (v_q))\right).
\end{align}

We now state the Lagrangian reduction theorem (\cite{CeHoMaRa1998,CeMaRa2001,ElGBHoRa2009}) in the case that external forces are allowed.
\medskip

\begin{thm}[Classical Lagrangian reduction with forces]\label{LPRed} Consider a curve $q(t)\in Q$ and define the two curves
\[
\rho(t):=\pi(q(t))\in \Sigma \quad\text{and}\quad  \bar{ \sigma} (t) := \bar {\mathcal{A}} (\dot q(t))\in \operatorname{Ad}Q,
\]
where $ \mathcal{A} $ is a fixed connection on $\pi:Q\rightarrow \Sigma$. Then, the following statements are equivalent:
\begin{enumerate}
\item Hamilton's variational principle,
\begin{equation}
\de \int^{t_2}_{t_1} L\lp q(t), \dot q(t)\rp\, dt  + \int^{t_2}_{t_1} F(q(t), \dot q(t))\de q(t)\, dt = 0,
\label{eq:varprin}
\end{equation}
holds for variations $\de q(t)$ of $q(t)$ vanishing at the endpoints.
\item The curve $q(t)$ satisfies the Euler-Lagrange equations with external forces:
\begin{align} \label{EL_forces}
\tilde{ \nabla }\dede{L}{\dot q} -\frac{\tilde\nabla L }{\delta q }  = F(q, \dot q),
\end{align} 
where $\tilde\nabla_t$ is the time covariant derivative associated to an arbitrary torsion free affine connection on $Q$ and $\frac{\tilde\nabla L}{\delta q }$ is the partial derivative of $L$ relative to the same connection. 
\item The constrained variational principle (of Lagrange-d'Alembert type)
\begin{align*} 
&\de \int^{t_2}_{t_1} \ell\lp\rho(t), \dot \rho(t),\bar 
\sigma (t)\rp\, dt\\
&\qquad + \int^{t_2}_{t_1}\left( \left\langle F^{\operatorname{Ad}Q}( \rho (t), \dot \rho (t), \bar{ \sigma }(t)), \bar \eta(t) \right\rangle  + \left\langle F^{\Sigma}( \rho (t), \dot \rho (t), \bar{ \sigma }(t)) , \de\rho(t) \right\rangle \right)  \,dt= 0
\end{align*}
holds, using variations $\de\rho(t)$ of $\rho(t)$ vanishing at the endpoints, and variations of $\bar\sigma (t)$ of the form
\[
\de\bar\sigma(t) = D_t\bar\eta(t) - \lsb \bar\sigma(t), \bar \eta(t)\rsb -\bar\cB\lp \dot \rho(t), \de\rho(t)\rp
\]
vanishing at the endpoints.
\item The Lagrange-Poincar\'e equations for $\ell$ with external forcing $F^{\operatorname{Ad}Q}$ and $F^{\Sigma}$ hold:
\begin{equation}\label{LP12}
\left\{\begin{array}{l}
\vspace{0.2cm}\displaystyle D_t\dede{\ell}{\bar \sigma} + \operatorname{ad}^*_{\bar \sigma}\dede{\ell}{\bar \sigma} = F^{\operatorname{Ad}}(\rho , \dot \rho , \bar{ \sigma })\\
\displaystyle\nabla_t\dede{\ell}{\dot \rho} - \frac{ \nabla \ell}{ \delta \rho} =F^\Sigma(\rho , \dot \rho , \bar{ \sigma }) -\scp{\dede{\ell}{\bar \sigma}}{{\mathbf{i} }_{\dot \rho}\bar \cB},
\end{array}\right.
\end{equation}
where $ \mathbf{i} _{ \dot{ \rho }} \bar{\mathcal{B}}$ denotes the $ \operatorname{ Ad}Q$-valued one-form on $\Sigma $ defined by
\[
\mathbf{i} _{ \dot{ \rho }} \bar{\mathcal{B}}:= \bar{\mathcal{B} }( \dot{ \rho }, \cdot ),
\]
$\nabla_t$ denotes the time covariant derivative relative to a torsion-free affine connection $ \nabla $ on $\Sigma$, and $ \frac{ \nabla \ell}{ \delta \rho}$ denotes the partial derivative of $\ell$ relative to the connection $ \nabla $.
\end{enumerate}
\end{thm}
\vspace{1em}

\begin{rema}{\rm
We have chosen to write the Euler-Lagrange equations \eqref{EL_forces} with the help of a torsion free connection $\tilde \nabla $ on $Q$, in order to have a global (i.e. coordinate independent) formulation. One can write as usual these equations locally as
\[
\frac{d}{dt} \dede{L}{\dot q} - \dede{L}{q} = F(q, \dot q),
\]
where $ \dede{L}{q}$ denotes the partial derivative of $L$ in a local chart.
}
\end{rema}

\begin{rema}[Energy and momentum map] {\rm In the presence of a force field,  the time derivative of the energy $E(q, \dot q)=\left\langle  \dot q, \frac{\partial  L}{\partial \dot q } \right\rangle - L(q, \dot q)$ along a solution of \eqref{EL_forces} is
\[
\frac{d}{dt} E( q(t), \dot q(t))= \left\langle F(q(t), \dot q(t)), \dot q(t) \right\rangle.
\]
A direct computation, using the same arguments as in \S 2.7 \cite{Ma1992},
shows that in the presence of forces, Noether's theorem is replaced by the relation
\begin{equation}\label{Noether_with_forces}
\frac{d}{dt} \mathbf{J} \left( \frac{\delta  L}{\delta \dot q } \right) =\mathbf{J} \left( F(q, \dot q) \right)=\mathbf{J} \left( F^v(q, \dot q) \right).
\end{equation} 
}
\end{rema}

\begin{rema}[Euler-Poincar\'e reduction] {\rm In the particular case $Q=G$, the $G$-invariant force field is completely determined by a smooth map $f: \mathfrak{g}  \rightarrow \mathfrak{g}  ^\ast $ and we have $\left\langle F(v_g), w_g \right\rangle = \left\langle f( v _g g ^{-1} ), w _g g ^{-1} \right\rangle$. In this case, the Lagrange-Poincar\'e equations \eqref{LP12} recover the Euler-Poincar\'e equations with force
\[
\frac{d}{dt} \frac{\delta \ell }{\delta \xi  } + \operatorname{ ad}^*_ \xi \frac{\delta \ell }{\delta \xi  } = f( \xi ).
\]
In this particular case, one observes that these equations are equivalent to Noether's formulation \eqref{Noether_with_forces}, since they can both be written as
\[
\frac{d}{dt} \operatorname{Ad}_{g} ^\ast \frac{\delta \ell }{\delta  \xi } =  \operatorname{Ad}_{g} ^\ast \left( f( \xi )\right) ,\quad \xi(t)= \dot g(t) g (t) ^{-1}.
\]
In the absence of forces, this recovers the usual fact that Euler-Poincar\'e equations are a direct consequence of Noether's theorem.
}
\end{rema}

\section{Distortions of projected dynamics}\label{sec:Dist}

In order to formulate an un-reduction procedure, we must relate the Lagrange-Poincar\'e equations \eqref{LP12} with the Euler-Lagrange equations (\ref{eq:EL}) on $\Sigma$ using the information obtained in Theorem \ref{LPRed}.  Here, differences between the Lagrange-Poincar\'e equations and the Euler-Lagrange equations are called {\em distortions}.  This terminology is suitable since, in the context of un-reduction, such differences are unwanted and the main goal of the un-reduction procedure is to eliminate them.

The conception of certain terms in the reduced equations as unwanted barriers to the objective stands in contrast with applications of reduction by symmetry, where such differences instead attain physical meaning.  Examples of this occur, for example, in \cite{ElGBHoPuRa2009} in the context of molecular strand dynamics where micro-structure of a strand couples with the reduced filament dynamics, or again in Kaluza-Klein constructions where curvature terms are understood as an electromagnetic field.


 There are two distortions to contend with:
\begin{enumerate}
\item {\bf Coupling distortion:} One often specifies a $G$-invariant Lagrangian $L:TQ \to \mR$ as opposed to a Lagrangian on $T\Sigma$, since coordinates on $\Sigma$ are often unavailable.  Since (\ref{eq:bundleiso}) shows that $TQ/G \cong T\Sigma \oplus\operatorname{Ad}Q$, the Lagrangian $L$ certainly depends on $\operatorname{Ad}Q$ if it is non-degenerate.  We must therefore remove any coupling between the $\operatorname{Ad}Q$ dependencies and the $T\Sigma$ dependencies from the unreduced Euler-Lagrange equations on $Q$.  In general, there is no preferred way to deal with these dependencies, and therefore this issue must be bourne in mind when setting up each particular un-reduction problem.  This ambiguity of the Lagrangian on $TQ$ is referred to hereafter as {\em coupling distortion}.

\item {\bf Curvature distortion:} The Lagrange-Poincar\'e equations \eqref{LP12} differ from the Euler-Lagrange equations (\ref{eq:EL}) in that the right hand side contains a driving term that arises from the curvature of the principal connection.  This driving term is referred to as {\em curvature distortion}.  We shall not discuss the nature of curvature distortion further, and point it out only for later reference when we develop methods for un-reduction in \S\ref{sec:strategies}.
\end{enumerate}

We now describe a particular class of $G$-invariant Lagrangians $L:TQ\rightarrow \mathbb{R}  $ well appropriate for the formulation of un-reduction. We say that the $G$-invariant Lagrangian {\em decouples relative to a connection $ \mathcal{A} $} if it can be written as the sum of two $G$-invariant Lagrangians, $L=L^h+L^v$, where $L^h:TQ\rightarrow \mathbb{R}$ and   $L^v:TQ\rightarrow \mathbb{R}$ are such that
\[
L^h(v_q)=L^h(P^h(v_q)) \quad\text{and}\quad L^v(v_q)=L^v(P^v(v_q)),\quad\text{for all $v_q\in TQ$},
\]
that is, $L^h(v_q)$ depends only on the horizontal part of $v_q$ and $L^v(v_q)$ depends only on the vertical part of $v_q$. Now we observe that the bundle isomorphism $\alpha _ \mathcal{A} :TQ/G \rightarrow T\Sigma\oplus\operatorname{Ad}Q$ induces two bundle isomorphisms
\[
\eval{\al_\cA}{VQ/G}:VQ/G \to \operatorname{Ad}Q \qquad \mbox{ and } \qquad \eval{\al_\cA}{HQ/G}:HQ/G  \to T\Sigma.
\]
Therefore, $L^h$ is completely determined by a Lagrangian $\ell^\Sigma:T\Sigma\to \mR$ and $L^v$ is completely determined by a function $\ell^{\operatorname{Ad}}:\operatorname{Ad}Q \to \mR$ through the relations
\[
L^h\left( \operatorname{Hor}_q(v_\rho ) \right) =\ell^ \Sigma ( v _\rho )\quad\text{and}\quad L^v \left(  \xi _Q (q) \right) =\ell^ {\operatorname{Ad}}\left( \llsb q,\xi \rrsb_{ \mathfrak{g}  }\right).
\]
Note also that
\begin{equation}\label{properties_fiber_der}
\frac{\delta  L^h}{\delta \dot q } = \pi^* \frac{\delta  \ell^ \Sigma }{\delta \dot \rho  } \in VQ^\circ\quad\text{and}\quad \frac{\delta  L^v}{\delta \dot q } =\bar{\mathcal{A}}^*\frac{\delta  \ell^{\operatorname{Ad}}}{\delta  \bar{\sigma}} \in HQ^\circ,
\end{equation}
where the last equality implies the relation $\frac{\delta  \ell^{\operatorname{Ad}}}{\delta  \bar{\sigma}} =\bar{ \mathbf{J} }\left(  \frac{\delta  L^v}{\delta \dot q }\right) $.

\begin{example}[Geodesic problems]\label{ex:geo}{\rm 
Let $\cG$ be a $G$-invariant Riemannian metric on $Q$, and consider the Lagrangian $L(v_q):=\frac{1}{2} \|v_q\|_{ \mathcal{G} }^2 $ associated to the metric. We will show that $L$ decouples relative to a connection and we will compute the induced functions $\ell^ \Sigma $ and $\ell ^ { \operatorname{ Ad}}$.

Recall that a $G$-invariant Riemannian metric naturally induces a connection defined by $H_qQ = \lp V_qQ\rp^\perp$ and called the mechanical connection. The associated connection form $ \mathcal{A} _{ \operatorname{mech}}$ is determined by the the relation
\[
\cG_q\lp v_q, q\cdot \sigma \rp = \scp{\mI(q)\cA_{\operatorname{mech}}\lp v_q\rp}{\sigma}_{\mg^*\times \mg},\quad\text{for all $ \sigma \in \mathfrak{g}$},
\]
where $ \mathbb{I}  (q):\mg \to \mg^*$ is the locked inertia tensor defined by
\[
\cG_q\lp q\cdot \nu, q\cdot \sigma\rp = \scp{\mI(q)\nu}{\mu}_{\mg^*\times\mg},\quad\text{for all $ \nu, \sigma \in \mathfrak{g} $}.
\]
The Lagrangian $L(v _q )=\frac{1}{2}\left\| v _q \right\|^2_\cG$ obviously decouples relative to the mechanical connection, since we can write
\[
L\lp v _q \rp = \frac{1}{2}\left\| v _q \right\|^2_\cG = \frac{1}{2}\left\|P_v\lp v _q \rp\right\|^2_{\cG} + \frac{1}{2} \left\|P_h\lp v _q \rp\right\|^2_{\cG}=L^v(v_q)+L^h(v_q).
\]
In order to write $\ell^{ \operatorname{ Ad}}$ explicitly, we need to introduce the vector bundle isomorphism $\cI:\operatorname{Ad}Q \to \operatorname{Ad}^* Q$  over the identity defined by
\[
\cI\lp \llsb q, \sigma \rrsb_{\mg}\rp = \llsb q, \mI(q)\sigma\rrsb_{\mg^*}.
\]
The induced functions $\ell^{\operatorname{Ad}}: \operatorname{Ad}Q\rightarrow \mathbb{R}  $ and $\ell^\Sigma: T \Sigma \rightarrow \mathbb{R}  $ are thus given by
\[
\ell^{\operatorname{Ad}}\lp \bar \sig\rp = \frac{1}{2}\scp{\cI\lp \bar \sig\rp}{\bar \sig} \qquad \mbox{ and } \qquad \ell^{\Sigma}\lp \rho, \dot \rho\rp = \frac{1}{2}\|\dot \rho\|^2_{\pi_*\cG},
\]
where $\pi_*\cG$ is the Riemannian metric on $\Sigma$ induced by $\cG$.$\qquad \blacklozenge$}
\end{example}

\paragraph{\bf Distorsion in the Lagrange-Poincar\'e equations.} Example \ref{ex:geo} shows that the class of geodesic Lagrangians is contained within the restricted class of Lagrangians on $Q$ that decouples relative to an appropriate choice of connection form. In this case, the reduced Lagrangian $\ell:TQ/G \to \mR$ takes the form $\ell = \ell^\Sigma + \ell^{\operatorname{Ad}}$.  Consequently, using the relations
\[
\frac{\delta  \ell}{\delta \rho  } =\frac{\delta  \ell^ \Sigma }{\delta \rho  }+\frac{\delta \ell^{\operatorname{ Ad}} }{\delta \rho  },\quad \frac{\delta  \ell}{\delta \dot \rho  }= \frac{\delta  \ell^ \Sigma }{\delta \dot \rho  },\quad  \frac{\delta  \ell}{\delta \bar \sigma   }=\frac{\delta  \ell^{\operatorname{ Ad}}}{\delta \bar \sigma},
\]
the Lagrange-Poincar\'e equations \eqref{LP12} read:
\begin{equation}\label{eq:LP12a}
\left\{\begin{array}{l}
\vspace{0.2cm}\displaystyle D_t \dede{\ell^{\operatorname{Ad}}}{\bar \sigma} - \ad^*_{\bar\sigma}\dede{\ell^{\operatorname{Ad}}}{\bar \sigma} = F^{{\operatorname{Ad}}}( \rho, \dot \rho , \bar \sigma)\\
\displaystyle\nabla_t\dede{\ell^\Sigma}{\dot \rho} - \frac{\nabla \ell^\Sigma}{\delta \rho} =F^\Sigma( \rho, \dot \rho , \bar \sigma  )  + \frac{\nabla \ell^{\operatorname{ Ad}}}{\delta \rho}  - \scp{\dede{\ell^{\operatorname{Ad}}}{\bar \sigma}}{\mathbf{i} _{\dot \rho}\bar\cB}.
\end{array}
\right.
\end{equation}

Here the third term on the right hand side of the second equation in \eqref{eq:LP12a} is the curvature distortion, while the second term gives an explicit form of the coupling distortion.  Note that the left hand side is the Euler-Lagrange operator on $\Sigma$ for the Lagrangian $\ell^\Sigma: T \Sigma \rightarrow \mathbb{R} $.  Thus, the goal of the un-reduction method is to selecting solutions to a Lagrangian system with forcing on $Q$ such that the right hand side of the second equation vanishes.  Such solutions in $Q$ project onto solutions of the Euler-Lagrange equations on $\Sigma$ with Lagrangian $\ell^\Sigma$.

\section{Un-reduction}\label{sec:strategies}

Having compared the Lagrange-Poincar\'e equations with the Euler-Lagrange equations on $T\Sigma$, we proceed by using the external force $F^\Sigma$ to cancel the right hand side of the second equation in \eqref{eq:LP12a}. Clearly we need to define the reduced force $F^ \Sigma:T \Sigma \oplus \operatorname{ Ad}Q\rightarrow T^* \Sigma $ by
\begin{equation}
F^\Sigma(\rho, \dot \rho , \bar{ \sigma }) : = \left\langle \frac{\delta  \ell^{\operatorname{Ad}}}{\delta   \bar \sigma},  \mathbf{i} _{\dot \rho}\bar\cB \right\rangle-\frac{\delta  \ell^{\operatorname{Ad}}}{\delta \rho} .
\label{eq:sigforce}
\end{equation}
For the moment we do not specify the other component $F^{ \operatorname{ Ad}}:T \Sigma \oplus \operatorname{ Ad}Q\rightarrow \operatorname{ Ad}^*Q$ of the reduced force. However, once $F^{ \operatorname{Ad}}$ is fixed, then the force field $F:TQ\rightarrow T^*Q$ is completely determined by the equality
\[
F(v_q)=F^h(v_q)+F^v(v_q)=\pi ^\ast \left( F^ \Sigma ( \alpha _ \mathcal{A} (v_q))\right)+\bar{ \mathcal{A} }^*\left( F^{ \operatorname{ Ad}}( \alpha _ \mathcal{A} (v_q))\right)
\]
as was shown in \eqref{relation_Fh_Fv}.

\subsection{The un-reduction theorem}

In order to compute the unreduced Euler-Lagrange equations associated to this forcing term $F$, we apply the variational principle \eqref{eq:varprin}. Using the same notations as in Theorem \ref{LPRed} we have
\begin{align*}
0&= \delta \int_0^1 L^h(q, \dot q) dt+\delta \int_0^1 L^v(q, \dot q) dt\\
& \qquad +\int_0^1 F^h(q, \dot q) \delta q dt+  \int_0^1 F^v(q, \dot q) \delta q dt\\
&= \delta \int_0^1 L^h(q, \dot q) dt+\int_0^1 \left\langle \frac{\delta \ell^{\operatorname{ Ad}} }{\delta\bar{ \sigma }} , \delta \bar \sigma \right\rangle + \left\langle \frac{\delta \ell^{\operatorname{ Ad}} }{\delta \rho } , \delta \rho \right\rangle dt\\
&\qquad +\int_0^1 \left\langle F^\Sigma (\rho , \dot \rho ,\bar{ \sigma }), \delta \rho  \right\rangle  dt+  \int_0^1 \left\langle F^{ \operatorname{ Ad}}(\rho , \dot \rho ,\bar{ \sigma }), \bar{ \mathcal{A} }( \delta q) \right\rangle  dt\\
&= \delta \int_0^1 L^h(q, \dot q) dt+\int_0^1 \left\langle \frac{\delta \ell^{\operatorname{ Ad}} }{\delta\bar{ \sigma }} , \delta \bar \sigma \right\rangle+\int_0^1 \left\langle \frac{\delta  \ell^{\operatorname{Ad}}}{\delta   \bar \sigma}, \bar\cB ( \dot \rho , \delta \rho )\right\rangle\\
&\qquad +  \int_0^1 \left\langle F^{ \operatorname{ Ad}}(\rho , \dot \rho ,\bar{ \sigma }), \bar{ \mathcal{A} }( \delta q) \right\rangle  dt,
\end{align*}
where we used \eqref{eq:sigforce} in the last equality.
Now the expression for the constrained variation $ \delta \bar{ \sigma }$ yields
\begin{align*}
\int_0^1 \left\langle \frac{\delta \ell^{\operatorname{ Ad}} }{\delta\bar{ \sigma }} , \delta \bar \sigma+ \bar\cB ( \dot \rho , \delta \rho )\right\rangle dt &=\int_0^1 \left\langle \frac{\delta \ell^{\operatorname{ Ad}} }{\delta\bar{ \sigma }} , D _t \bar{\eta}-[\bar{ \sigma }, \bar{\eta}]\right\rangle dt\\
&= \int_0^1 \left\langle \bar{ \mathbf{J} } \left( \frac{\delta L^v}{\delta\dot q }\right)  , \llsb q(t), \dot{\eta}(t)\rrsb_{\mg^*}\right\rangle dt \\
&=\int_0^1 \left\langle \mathbf{J} \left( \frac{\delta L^v}{\delta\dot q }\right)  ,\dot{\eta}(t)\right\rangle dt\\
&= -\int_0^1 \left\langle \frac{d}{dt} \mathbf{J} \left( \frac{\delta L^v}{\delta\dot q }\right), \mathcal{A} ( \delta q)\right\rangle dt,
\end{align*}
since $\bar{\eta}(t)= \llsb q(t), \mathcal{A} ( \delta q(t)) \rrsb_{ \mathfrak{g}  }$. Therefore, the variational principle yields the Euler-Lagrange equations
\begin{align}\label{EL_h_v}
\frac{d}{dt} \frac{ \delta  L^h}{ \delta  \dot q}-  \frac{ \delta    L^h}{ \delta q}= \bar { \mathcal{A} }^* F^{ \operatorname{ Ad}}( \rho , \dot \rho , \bar \sigma )- \mathcal{A} ^* \frac{d}{dt} \mathbf{J} \left(\frac{ \delta  L^v}{ \delta  \dot q} \right) .
\end{align} 
These equations may be split into horizontal and vertical parts by observing that  the terms on the right hand side belong to $HQ^\circ$ and the terms on the left hand side belong to $VQ^\circ$ since we have
\begin{equation} \label{derivative_Lh}
\frac{\delta L ^h  }{\delta  \dot q} =\pi^* \frac{\delta \ell ^\Sigma  }{\delta  \dot \rho } \quad\text{and}\quad \frac{\delta L ^h  }{\delta   q} =\pi^* \frac{\delta \ell ^\Sigma  }{\delta  \rho }.
\end{equation} 
Thus, the un-reduction equations \eqref{EL_h_v} can be equivalently written as
\begin{eqnarray}\label{eq:un1}
&&\frac{d}{dt}  \dede{L^h}{\dot q} - \dede{L^h}{q}=0\\
&&\frac{d}{dt}  \mathcal{A} ^* \mathbf{J} \left(\dede{L^v}{\dot q} \right)  = F^v(q, \dot q),
\label{eq:un3}
\end{eqnarray}
where $F^v: TQ\rightarrow HQ^\circ$ is completely determined by $F^{ \operatorname{Ad}}$.
Note that by the first equality in \eqref{derivative_Lh}, we have the relation
\[
\mathbf{J} \left(\dede{L^h}{\dot q}  \right)=0.
\]
Note also that equation \eqref{eq:un3} can be rewritten as
\begin{equation}
\frac{d}{dt}  \mathbf{J} \left(\dede{L^v}{\dot q} \right)  = \mathbf{J} \left( F^v(q, \dot q) \right) ,
\label{eq:AltVert}
\end{equation}
since $F^v(q, \dot q)\in HQ^ \circ$. This equation can also be deduced from the relation \eqref{Noether_with_forces}.

\rem{
A simple example of a force $F^v$ is obtained by choosing a $G$-equivariant map $ \gamma :Q\rightarrow \mathfrak{g}  ^\ast $ and defining
\[
F^v(v_q):=\mathcal{A} ^*(\gamma (q))\in HQ^\circ.
\]
The $G$-equivariance is checked as follows:
\[
\left\langle F^v(v_q \cdot g), w_q \cdot g \right\rangle = \left\langle \gamma ( q \cdot g), \mathcal{A} ( w _q \cdot g) \right\rangle = \left\langle \operatorname{Ad}^\ast _g \gamma (q), \operatorname{Ad}_{g ^{-1}}\mathcal{A}(w _q ) \right\rangle=\left\langle F^v(v_q), w_q \right\rangle
\]
For this special choice, equation \eqref{eq:un3} reads simply
\[
\frac{d}{dt}  \mathbf{J} \left(\dede{L^v}{\dot q} \right)  = \gamma (q).
\]
}


The results obtained so far are summarized in the following theorem:

\begin{thm}[Un-reduction]\label{thm:un} Let $\ell^ \Sigma :T \Sigma \rightarrow \mathbb{R}  $ be an arbitrary Lagrangian defined on the base of a principal bundle $\pi: Q\rightarrow Q/G=\Sigma$. Let $L=L^h+L^v$ be a $G$-invariant Lagrangian such that $L^h\circ P_h=L^h$ and $L^v\circ P_v=L^v$, where $L^h$ is uniquely determined by $\ell^ \Sigma $, and choose an arbitrary force $F^v:TQ\rightarrow HQ^\circ$.

Then, solutions $q(t)$ of the  equations \eqref{eq:un1} - \eqref{eq:un3} project to solutions $ \rho (t):=\pi(q(t))$ of the Euler-Lagrange equations for $\ell^\Sigma$.
\end{thm}

\paragraph{\bf Remark.} We recall here that equations \eqref{eq:un1} - \eqref{eq:un3} are the unreduced Euler-Lagrange equations on $TQ$ for a Lagrangian $L=L^v+L^h$ and a force field $F=F^v+F^h$, in the special case when $F^h$ is constructed from the given Lagrangian  $\ell^\Sigma $ by the formula \eqref{eq:sigforce}. Note that once a connection has been fixed, the choice of $\ell^\Sigma$ determines $L^h$ and vice versa.  However, the choice of $L^v$ and $F^v$ is left open.  The main content of un-reduction is that the projected dynamics on $\Sigma$ is independent of the choice of $L^v$ and $F^v$. Therefore, both these functions may be chosen arbitrarily, a freedom that constitutes a distinct modeling step for the particular application at hand.  In practice, vertical dynamics, that is equation \eqref{eq:un3}, may be chosen
 to give suitable numerical properties or to add additional functionality to the solution being developed.  An example of such a modeling procedure is given in \S\ref{PlaneCurves}.

Theorem \ref{thm:un} makes no statement about the existence of solutions for the equations \eqref{eq:un1} - \eqref{eq:un3}. Usually, the simplest way is to remark that the un-reduction equations are equivalent to Euler-Lagrange equations \eqref{EL_forces} for the Lagrangian $L = L^h + L^v$.

Also of interest is the converse of Theorem \ref{thm:un}. Given a curve $\rho(t)$ that solves the Euler-Lagrange equations for the Lagrangian $\ell^\Sigma$ on the base space, can we find a solution $q(t)$ of \eqref{eq:un1} - \eqref{eq:un3}, which projects down to $\rho(t)$? In the case of geodesic problems as in Example \ref{ex:geo} we know from differential geometry that this is always possible. Given a geodesic on $\Sigma$ and a Riemannian submersion $\pi:Q\to\Sigma$, we can lift it to a horizontal geodesic on $Q$. In the general case the possibility of lifting will depend on whether there exist solutions to \eqref{eq:un3}.

\subsection{Relationship with vanishing momentum maps and horizontal methods for geodesic problems}\label{sec:geo}


The horizontal methods for geodesic problems, developed and applied in \cite{MiMu07,CoHo2009}, may be recovered from the general un-reduction procedure developed here.  The horizontal approach takes the Euler-Lagrange equations on $Q$ with the geodesic Lagrangian and asserts, in addition, the condition
\begin{equation}
\mathbf{J}\lp \dede{L}{\dot q}\rp = 0.
\label{eq:horcond}
\end{equation}

To gain a clearer grasp of the horizontal approach for geodesic problems, consider first the more general problem of asserting \eqref{eq:horcond} for a general Lagrangian $L=L^h+L^v$. Since $\delta L^h/\delta \dot q \in VQ^\circ$, the relation $\mathbf{J}\left(\delta L^h/\delta \dot q\right)=0$ is always satisfied, therefore \eqref{eq:horcond} implies $\mathbf{J}\left(\delta L^v/\delta \dot q\right)=0$ and we have $\de \ell^{\operatorname{Ad}} / \de \bar \rho = \bar {\bf J}\lp \de L^v / \de \dot q\rp=0$. This reduces the Lagrange-Poincar\'e equations \eqref{eq:LP12a} with zero external forces to the following single equation,
\[
\nabla_t\dede{\ell^\Sigma}{\dot \rho} - \frac{\nabla \ell^\Sigma}{\delta \rho} = \frac{\nabla \ell^{\operatorname{ Ad}}}{\delta \rho}.
\]

At the un-reduced level, since we already know that $\delta L^v/\delta \dot q\in HQ^\circ$ by \eqref{properties_fiber_der}, the condition $\mathbf{J}\left(\delta L^v/\delta \dot q\right)=0$ implies that $\de L^v / \de \dot q = 0$. Consequently, the Euler-Lagrange equations in the vanishing momentum case take the form
\[
\tilde \nabla_t\dede{L^h}{\dot q} - \frac{\tilde\nabla L^h}{\de q} = \frac{ \tilde \nabla L^v}{\de q}.
\]

Therefore, the vanishing momentum map condition eliminates curvature distortion, although coupling distortion may persist.

Note that in general these equations do not coincide with the un-reduction equations with the vanishing momentum map condition imposed. Indeed, from \eqref{eq:un1} - \eqref{eq:un3} we obtain that these un-reduced equations read
\begin{eqnarray}\label{eq_vanmom1}
&&\tilde \nabla_t \dede{L^h}{\dot q} - \frac{\tilde\nabla L ^h }{\delta q } = 0\\
&&\mathbf{J}\lp \dede{L}{\dot q}\rp = 0,
\label{eq_vanmom2}
\end{eqnarray}
where the coupling distortion term $\tilde \nabla L^v/\de q$ has now been cancelled by the horizontal forcing.

Specializing to the particular case when $L$ is the geodesic Lagrangian and $\cA$ is the corresponding mechanical connection, the vanishing momentum map condition \eqref{eq:horcond} becomes a horizontality condition,
\[
\left( \frac{\delta  L}{\delta  \dot q}\right) ^\sharp = \dot q \in HQ.
\]
In addition, for the geodesic Lagrangian, $\tilde \nabla L^v / \de q = \tilde \nabla L / \de q = 0$ where $\tilde \nabla$ is the Levi-Civita connection.  Thus, in the case of geodesic motion with vanishing momentum, the un-reduction equations are equivalent to the horizontal geodesic equations themselves
\begin{eqnarray}
\tilde \nabla_t \dot q &=& 0 \label{eq:unGeo1}\\
{\bf J}\lp \dede{L}{\dot q}\rp &=& 0.\label{eq:unGeo2}
\end{eqnarray}
Whilst \eqref{eq:unGeo1} - \eqref{eq:unGeo2} may appear over-determined, the observation that \eqref{eq:unGeo1} is an Euler-Lagrange equation with symmetry allows one to interprete \eqref{eq:unGeo2} as Noether's Theorem.  Therefore \eqref{eq:unGeo1} - \eqref{eq:unGeo2} are consistent, and \eqref{eq:unGeo2} asserts a constraint on the initial conditions.  Namely, $\cA\lp \dot q_0\rp = 0$, that is the initial velocity must be horizontal.  Thus, un-reduction applied to the geodesic problem together with the vanishing momentum map condition recovers to the horizontal shooting method for geodesic problems.

This coincidence occurs because, in addition to the property of vanishing momentum to eliminate curvature distortion, the particular choice of the geodesic Lagrangian does not introduce any coupling distortion.  Therefore, the forcing (both horizontal and vertical) are equal to zero, and the un-reduction procedure drops to reconstruction of a symmetry reduced Lagrangian system.  That is, the un-reduction procedure tells us to seek solutions of the original Euler-Lagrange equations (i.e. the geodesic equation) that satisfy the vanishing momentum map condition.

As we have seen above, however, for a general Lagrangian the vanishing momentum map condition is not admissible without modification since it does not allow us to obtain the Euler-Lagrange equations for $\ell ^\Sigma $ on $ \Sigma $. In this case, one needs to use the un-reduction approach developed here as opposed to simply asserting the vanishing momentum condition.

\section{Closed, simple plane curve matching}\label{PlaneCurves}

This section develops equations for the problem of matching closed, simple, planar curves using the un-reduction techniques developed in the paper.  The matching problem for such curves has been treated in, for example, \cite{Co2008,CoHo2009,MiMu07,MiMu2006}.  The formulation of the problem differs between authors, however, the methods developed achieve consensus in agreeing that the path followed by the matching algorithms should be geodesics.  In this section, we present an alternative perspective.  Since any observer of the curve dynamics is only sensible of the geometric information, or `shape', of a curve, we propose that matching requires only that shape need follow a geodesic.  Meanwhile, the dynamics of the whole curve, which includes information about both parametrization and shape, may evolve in a non-geodesic fashion.

This broader notion of curve matching allows the introduction of an entire family of matching dynamics, one of whom is the original geodesic dynamics.  This new family of matching dynamics, which is described efficiently by the un-reduction technique outlined in this paper, incorporates enough flexibility to implicitly overcome other difficulties.  Here, for example, we use the extra flexibility brought by un-reduction matching to address the unwieldy parametrization dynamics displayed by geodesic matching, as in Figure $1$ in \cite{CoHo2009}, by giving dynamics that remain uniformly parametrized throughout the matching procedure.  This achieves the objective set out in \S\ref{sec:Intro} to derive a set of matching dynamics that addressed the parametrization problems faced by curve matching algorithms in an intrinsic, dynamic way along the entire trajectory. 

Additionally, we present a member of the un-reduction family whose dynamics respect the property of curves being written as a graph in polar coordinates, wherever it is sensible to do so.  In higher dimensional problems, such as that studied in \cite{BaHaMi2010}, the graph preservation property may significantly simplify numerical implementation of the dynamics.  These two examples of non-geodesic matching dynamics will, hopefully, stimulate the description of matching dynamics with other additional properties.

\subsection{Geometric setup}

Following \cite{MiMu2006}, consider $Q=\Emb^+\lp S^1, \mR^2\rp$, the smooth manifold of all positively oriented embeddings from $S^1$ to $ \mathbb{R}^2$.  $Q$ may be thought of as the space of simple, closed, planar curves.  An element $c \in Q$ contains information about a shape, namely a submanifold of $\mR^2$ of $S^1$-type, and a parametrization of the shape.   The space of shapes may therefore be identified with the quotient manifold $\Sigma = Q/G$, where $G=\operatorname{Diff}^+(S^1)$ is the group of orientation preserving diffeomorphisms of $S^1$ which acts freely and properly from the right on $Q$ by composition. We refer to \cite{KrMi1997} for a detailed discussion of the manifold structure of $\Emb\left(  S^1, \mathbb{R}  ^2  \right) $ and $ \Sigma $, and to \cite{MiMu2006} for a discussion on the geometry of spaces of immersions and their quotients.

In particular, in an appropriate topology, $\Emb^+\left(  S^1, \mathbb{R}  ^2  \right)$ is an open subset of $ C^\infty (S^1, \mathbb{R}  ^2 )$. The quotient manifold, however, is not feasible for numerical implementation since there are no natural coordinates and analytical considerations can only be achieved via the use of equivalence classes.  Therefore, it is natural to resort to an un-reduction procedure that takes advantage of the principal $\Diff^+\lp S^1\rp$-bundle structure of $Q$ over $\Sigma$.

Recall that an embedding $c:S^1 \rightarrow \mathbb{R} ^2 $ is a smooth injective immersion that is a homeomorphism onto its image. Thus, $Q$ may be identified with smooth, periodic maps $C^\infty_{1}\lp \mR;\mR^2\rp$, with, say, period $1$, that are injective on each interval $[a,1+a)$ for all $a \in \mR$, and satisfy $|c_\theta|\lp \theta\rp \neq 0$ for all $\theta \in \mR$. Here the subscript denotes differentiation.

Similarly, an element $f \in G:=\Diff^+(S^1)$ may be identified with a smooth, strictly monotonic map $f: \mR \to \mR$ such that $f\lp \theta + 1 \rp = f\lp \theta\rp + 1$ for all $\theta \in \mR$. 

Since $Q$ is an open subset of $ C^\infty (S^1, \mathbb{R}  ^2 )$, a tangent vector $U_c \in T_cQ$ is represented by a pair of maps $(c,U)$, where $U\in C^\infty _{1}(\mR, \mathbb{R}  ^2 )$. The Lie algebra of $G$ is given by the space $\mg:=\mX\lp S^1\rp$ of vector fields on $S^1$, which may also be represented by real valued, periodic functions $C^\infty_1\lp \mR, \mathbb{R}  \rp$. The infinitesimal generator associated to $u\in \mg$ reads $u_ Q\lp c\rp  = \lp c, u c_\theta\rp$.

The following are some geometric quantities which are convenient to use due to their behavior under the action of $G$.
\begin{itemize}
	\item Derivative along the curve
	\[
	D_\theta = \frac{1}{|c_\theta|}\prt_\theta
	\]
	\item Length of the curve
	\[
	l(c)= \int^1_0|c_\theta| \, d\theta
	\]
	\item Unit tangent vector
	\[
	\tau (c)= \frac{c_\theta}{|c_\theta|}, \textrm{ or } \tau (c) = D_\theta c
	\]
	\item Unit normal vector
	\[
	n(c) = J\tau(c), \mbox{ where } J =
	\begin{pmatrix}
	0 & -1\\
	1 & 0\\
  \end{pmatrix}
	\]
	\item Curvature
	\[
	\kappa(c) = \lp D_\theta\tau(c)\rp\cdot n(c)
	\]
	\item Volume measure
	\[
	\operatorname{vol}(c) = |c_\theta|\, d\theta
	\]
\end{itemize}

A cotangent vector in $T^*_cQ$ is represented by a pair of maps $(c,P\otimes\Omega)$, where $P\otimes\Omega\in C_1^\infty(\mathbb{R}, (\mathbb{R}^2)^*) \otimes \Omega^1(S^1)$. The dual space to the Lie algebra $\mg = \mathfrak{X} (S^1)$ is identified with $ \mg^* = \mathfrak{X} (S^1)^* = C_1^\infty(\mathbb{R}, \mathbb{R}) \otimes \Omega^1(S^1)$. This is seen by introducing the following pairings
\[
\left\langle (c,P\otimes \Omega), (c,U) \right\rangle := \int_{S^1} P( U) \, \Omega,\quad\text{and}\quad \left\langle \mu\otimes\omega,u \right\rangle =\int_{ S^1 } \mu (u) \, \omega,
\]
where $P\otimes \Omega \in T^*_cQ$, $(c,U) \in T_cQ$, $\mu \otimes \om \in \mg^*$, and $u \in \mg$.

The cotangent lift momentum map, $\mathbf{J}:T^*Q \to \mg^*$, that arises with this setup is calculated to be
\[
\mathbf{J}\lp c,P\otimes\Omega \rp = (P\cdot  c_\theta)\otimes \Omega.
\]

Fixing $\Om = \operatorname{vol}(c)$ to be the volume measure on $c$, the momentum map becomes
\begin{equation}
\mathbf{J}\lp c,P\otimes \operatorname{vol}(c) \rp = |c_\theta|(P\cdot  \tau(c))\otimes \operatorname{vol}(c)\in\mathfrak{g}^*.
\label{eq:mommap}
\end{equation}

The simplest metric to consider would be the $L^2$ metric. However, as pointed out in \cite{MiMu2006}, this metric has arbitrarily small geodesic distance between any two curves.  The simplest metric that is of practical importance is therefore the $L^2$ metric weighted by curvature.  That is,
\begin{equation}
\mathcal{G} _c \left( U,V \right) = \int^{1}_0 \lp1+ \kappa^2(c)\rp \lp U\cdot V\rp\, \operatorname{vol}(c).
\label{eq:Metric}
\end{equation}
For a detailed discussion of the properties of this and other metrics such as Sobolev metrics, see \cite{MiMu07,MiMu2006}.
The mechanical connection corresponding to $\cG$ is given by
\begin{equation}
\mathcal{A} _{ \operatorname{mech}}\lp(c,U)\rp = \frac{U\cdot \tau(c)}{|c_\theta|}.
\label{eq:connection}
\end{equation}
 
Indeed, since 
\[\mathcal{G} _c(u_Q(c), v_Q(c))
=
\int_0^{1} \lp 1+\kappa^2(c)\rp u v |c_\theta|^2\, \operatorname{vol}(c)
,
\]
the locked inertia tensor is $ \mathbb{I}(c)u= \lp 1+\kappa^2(c)\rp |c_\theta|^2 u \otimes \operatorname{vol}(c)$.   Therefore, the identity 
\[
\left\langle \mathbb{I}(c) \mathcal{A} _{ \operatorname{mech}}(c,U), u \right\rangle = \mathcal{G} _c\left( (c,U), (c,u c_\theta)\right)
\]
implies formula \eqref{eq:connection}. The projections associated with $\cA_{\operatorname{mech}}$ are
\[
P^v(c,U)=(c,(U\cdot \tau(c)) \tau(c)) \quad\text{and} \quad P^h(c, U)=(c,U- (U\cdot \tau(c)) \tau(c)) = (c,(U\cdot n(c))n(c)).
\]
Thus, the horizontal-vertical split of a vector field along $c$ is just its decomposition into normal and tangent components.

Let $L$ be the geodesic Lagrangian corresponding to (\ref{eq:Metric}),
\begin{equation}\label{lagrangian_kappa}
L(c, c_t) = \frac{1}{2} \|c_t\|^2_Q= \frac{1}{2}\int_0^{ 1} \lp 1+ \kappa^2(c) \rp | c_t|^2\, \operatorname{vol}(c).
\end{equation}
The horizontal part of $L$ is
\begin{equation}\label{eq:horLagrangian}
L^h (c,c_t)= \frac{1}{2}\int^1_0 \lp 1 + \kappa^2(c)\rp h(c,c_t)^2 \, \operatorname{vol}(c)
\end{equation}
where the notation $h(c,c_t) :=c_t\cdot n(c)$, so that $P^h(c,c_t) = h(c,c_t)n(c)$.
   
\subsection{Derivation of un-reduction equations}   

In order to derive the un-reduction equations \eqref{eq:un1} and \eqref{eq:AltVert} for the planar curve matching problem one must write down the Euler-Lagrange equations for the Lagrangian $L^h$. This is best accomplished by using a variational principle.  For the variational approach, one must first calculate the variations that will be used in the variational principle before calculating the equations.  This calculation is achieved by the following lemma, after which un-reduction equations are derived.

\begin{lem}\label{lem:var}
The variations $\de h$, $\de\operatorname{vol}(c)$ and $\de \kappa$ are
\begin{eqnarray*}
\de h &=& \lp \prt_t - sD_\theta\rp \lp \de c \cdot n\rp + D_\theta h \lp \de c \cdot \tau\rp\\
\de\operatorname{vol}(c) &=& \lp D_\theta\lp \de c\cdot \tau\rp - \lp \de c \cdot n\rp\kappa \rp\, \operatorname{vol}(c),\\
\de \kappa &=& \lp D^2_\theta  + \kappa^2\rp \lp \de c \cdot n\rp + D_\theta \kappa \lp \de c \cdot \tau\rp.
\end{eqnarray*}
where for simplicity we write $h$, $n$, $\tau$, $\kappa$ instead of $h(c,c_t)$, $n(c)$, $\tau(c)$, $\kappa(c)$.
\end{lem}

\begin{proof}
It is convenient to calculate the commutator of differential operators, $[\de,D_\theta]$, before continuing.  One has
\begin{eqnarray}
[\de, D_\theta] &=& \de \frac{\prt_\theta}{|c_\theta|} - \frac{\prt_\theta}{|c_\theta|} \de \nonumber\\
								&=& -\frac{c_\theta \cdot \de c_\theta}{|c_\theta|^3}\prt_\theta \nonumber\\
								&=& -\lp \tau \cdot D_\theta \de c\rp D_\theta\nonumber\\
								&=& \lp \lp n\cdot \de c\rp \kappa- D_\theta\lp \tau\cdot \de c\rp\rp D_\theta,\label{eq:Commutator}
\end{eqnarray}
where we used $D_\theta\tau=\kappa n$.
Employing \eqref{eq:Commutator}, one may precalculate the variation of the tangent field, $\tau$,
\begin{eqnarray}
\de \tau 	&=& \de D_\theta c\nonumber\\
 					&=& D_\theta \de c  + [\de, D_\theta] c\nonumber\\
 					&=& D_\theta \lp \lp \tau \cdot \de c\rp \tau + \lp n \cdot \de c\rp n\rp + \lp \lp n \cdot \de c\rp \kappa- D_\theta\lp \tau \cdot \de c\rp \rp\tau\nonumber\\
 					&=& \lp \lp \tau\cdot \de c\rp\kappa + D_\theta\lp n\cdot \de c\rp \rp n.
\label{eq:varTang}
\end{eqnarray}
where we used $D_\theta n=-\kappa \tau$. Similarly, for the time derivative we have
\[
\partial_t\tau=\lp \lp \tau\cdot  c_t\rp\kappa + D_\theta\lp n\cdot c_t\rp \rp n=\lp \lp \tau\cdot  c_t\rp\kappa + D_\theta h \rp n.
\]
Multiplying these variations by $J$ we obtain for $n$
\[
\de n = J\de \tau = - \lp D_\theta\lp \de c \cdot n\rp + \lp \de c \cdot \tau\rp\kappa \rp \tau\quad\text{and}\quad \partial_t n = J\partial_t \tau = - \lp D_\theta h + s\kappa \rp \tau,
\]
where we used the notation $s:=c_t\cdot\tau$.

Proceeding with the expression for $\de h$, one has
\begin{eqnarray*}\label{eq:hvar}
\de h &=&\de \lp c_t\cdot n\rp = \de c_t \cdot n + c_t \cdot \de n\\
&=&\prt_t\lp \de c \cdot n \rp - \de c \cdot \partial_tn- \lp D_\theta\lp \de c \cdot n\rp + \lp \de c \cdot \tau\rp\kappa \rp s\\
&=&\lp \prt_t - sD_\theta\rp \lp \de c \cdot n\rp + D_\theta h\lp \de c \cdot \tau\rp
\end{eqnarray*}
as required.

The expression for $\de \operatorname{vol}(c)$ follows since
\begin{eqnarray*}
\de \operatorname{vol}(c) &=& \frac{\de c_\theta \cdot c_\theta}{|c_\theta|}\, d\theta\\
&=& \lp D_\theta \de c\cdot \tau\rp\, \operatorname{vol}(c)\\
&=& \lp D_\theta\lp \de c\cdot \tau \rp - \lp \de c \cdot n\rp\kappa\rp \, \operatorname{vol}(c).
\end{eqnarray*}
Finally, upon noting that $\de \tau \propto D_\theta \tau \propto n$, and similarly reversing $\tau$ and $n$, the derivation of the expression for $\de \kappa$ reads
\begin{eqnarray*}
\de \kappa 	&=& \de\lp D_\theta \tau \cdot n\rp\\
						&=& D_\theta \de \tau \cdot n + [\de, D_\theta]\tau \cdot n + D_\theta \tau \cdot \de n\\
						&=& D_\theta\lp \de \tau \cdot n\rp + [\de, D_\theta]\tau \cdot n\\
						&=& D_\theta\lp\lp \tau\cdot \de c\rp \kappa + D_\theta\lp n\cdot \de c\rp\rp + \kappa\lp \lp n\cdot \de c\rp\kappa - D_\theta\lp \tau\cdot \de c\rp\rp\\
						&=& \lp D^2_\theta + \kappa^2\rp\lp \de c \cdot n\rp +  D_\theta\kappa \lp \de c \cdot \tau\rp,
\end{eqnarray*}
as required. 
\end{proof}

\begin{prop}\label{prop:unvar}
The un-reduction equations for simple, planar, closed curves with the geodesic Lagrangian associated with $\cG$ read
\begin{eqnarray*}
h_t &=& D_\theta \lp s h\rp - \frac{\kappa\lp 1+ 3\kappa^2\rp}{2\lp 1 + \kappa^2\rp}h^2 + \frac{1}{\lp 1+\kappa^2\rp}\lp D^2_\theta\lp \kappa h^2\rp - 2 \kappa h D^2_\theta h\rp\\
\\
s_t &=& f\lp c, s, h\rp\\
\\
c_t &=& hn + s\tau.
\end{eqnarray*}
where $f$ is an arbitrary smooth, $G$-invariant map,
\[
f: \Emb\lp S^1, \mR^2\rp \times C^\infty_1\lp \mR\rp \times C^\infty_1\lp \mR\rp \to C^\infty_1\lp \mR\rp.
\]
\end{prop}

\begin{proof}

The first step is to calculate the variational principle for Lagrangian \eqref{eq:horLagrangian} using the variations derived in Lemma \ref{lem:var}.  The variational principle may be calculated as
\begin{eqnarray*}
0 &=& \de \int^1_0 \int^1_0 \frac{1}{2}\lp 1 + \kappa^2\rp h^2 \, \operatorname{vol}(c)\, dt\\
	&=& \int^1_0 \int^1_0 \lp 1 + \kappa^2\rp h\de h \operatorname{vol}(c)\, dt + \int^1_0 \int^1_0\frac{1}{2}\lp 1 + \kappa^2\rp h^2 \, \de\operatorname{vol}(c)\, dt\\
	&& \quad + \int^1_0 \int^1_0 \kappa h^2 \de\kappa \, \operatorname{vol}(c)\, dt.
\end{eqnarray*}
After substituting the variations from Lemma \ref{lem:var} and integrating by parts, the terms proportional to $\lp \de c \cdot \tau\rp$ vanish as expected since the Lagrangian $L^h$ is known to be degenerate.  The remaining terms, those proportional to $\lp \de c \cdot n\rp$, yield the equation
\begin{equation}
h_t = D_\theta\lp s h\rp - \frac{\kappa h^2}{2} + \frac{1}{1+\kappa^2} \lp \lp D^2_\theta + \kappa^2\rp\lp \kappa h^2\rp - h\lp \prt_t - sD_\theta\rp \lp 1+ \kappa^2\rp\rp.
\label{eq:hun-red}
\end{equation}
Now, replacing $\de$ with $\prt_t$ in the expression for $\de \kappa$ from Lemma \ref{lem:var} yields the following result
\[
\lp \prt_t - sD_\theta\rp \kappa = \lp D^2_\theta + \kappa^2\rp h.
\]
Employing this relation on \eqref{eq:hun-red} one obtains
\[
h_t = D_\theta\lp s h\rp - \frac{\kappa h^2}{2} + \frac{1}{1+\kappa^2} \lp \lp D^2_\theta + \kappa^2\rp\lp \kappa h^2\rp - 2\kappa h\lp D^2_\theta + \kappa^2\rp h\rp.
\]
Rearranging terms yields the desired result.

The equation for $c$ is simply a decomposition of $c_t$ into horizontal and vertical parts.  That is,
\begin{equation}
c_t = P^h(c_t) + P^v(c_t) = hn + s\tau.
\label{eq:Decomp}
\end{equation}

Next, for the equation for $s$, note that \eqref{eq:mommap} implies the relation
\begin{equation}\label{momap_geodesic}
\mu = {\bf J}\lp \dede{L}{c_t}\rp = |c_\theta|\lp 1+\kappa^2\rp s\otimes \operatorname{vol}(c)=\mathbb{I}(c)s
\end{equation}
so equation 
\eqref{eq:AltVert} reads
$\frac{d}{dt} \left(\mathbb{I}(c)s\right) = {\bf J}\lp F^v\lp c, c_t\rp\rp.$ Computing the time derivative on the left hand side and rearranging the terms yields the relation $s_t=\mathbb{I}(c)^{-1}\left({\bf J}\lp F^v\lp c, c_t\rp\rp-\frac{d}{dt}(\mathbb{I}(c)) s\right)$. Substituting in equation \eqref{eq:Decomp} on the right hand side yields an equation for $s$ that may be expressed as
\[
s_t = f\lp c, s, h\rp,
\]
for an arbitrary smooth, $G$-invariant map $f$.  Of course, there is a relationship between $f$ and ${\bf J}\lp F^v\lp c, c_t\rp\rp$, however, this relationship is not necessary for the present purposes, so it has been omitted.

\end{proof}

The un-reduction equations for curve matching derived in Proposition \ref{prop:unvar} represent a family of dynamical equations parametrized by an exogenous choice of smooth function $f:\mR^2\times \mR^2 \to \mR$.  For any initial conditions $\lp c_0, \dot c_0\rp \in TQ$ and any choice of exogenous function $f$, the projected dynamics in $Q/\Sigma$, obtained by taking the image of a curve $c$, are guaranteed to follow geodesics under the metric induced on shape space by the curvature-weighted $L^2$ metric on $Q$.  Consequently, the choice of parametrization of the boundary data in matching algorithms, and the choice of forcing may be freely chosen for convenience of the user, to pose an optimal control problem, or to introduce extra modeling into the dynamics without needing to re-derive the equations.  We present a couple of uses of this extra flexibility in the remainder of this section.  It should be noted, however, that the examples we give only scratch the surface of possible uses of the extra functionality brought by un-reduction.

\subsection{Particular choices of vertical forcing and initial conditions}

\subsubsection{Horizontal geodesics}
The horizontal geodesics case has already been described in some detail in \S\ref{sec:geo}. Recall that $c$ is horizontal if $s=0$ (consistently with the vanishing momentum map condition \eqref{momap_geodesic}). This case corresponds to the choice $f = 0$ together with the initial condition $s_0 = 0$, and recovers the horizontal shooting method of \cite{CoHo2009,MiMu07}.  The equations for this case are given by
\begin{eqnarray}\label{eq:hshoot1}
h_t &=& - \frac{\kappa\lp 1+ 3\kappa^2\rp}{2\lp 1 + \kappa^2\rp}h^2 + \frac{1}{\lp 1+\kappa^2\rp}\lp D^2_\theta\lp \kappa h^2\rp - 2 \kappa h D^2_\theta h\rp\\
\nonumber \\
c_t &=& hn.\label{eq:hshoot2}
\end{eqnarray}

Note that by choosing $f=0$, but allowing a free choice of the initial condition $s_0$, one recovers the method introduced in \cite{CoHo2009} to address parametrization concerns.  These equations read
\begin{eqnarray}\label{eq:CHhshoot1}
h_t &=& D_\theta\lp s_0 h\rp - \frac{\kappa\lp 1+ 3\kappa^2\rp}{2\lp 1 + \kappa^2\rp}h^2 + \frac{1}{\lp 1+\kappa^2\rp}\lp D^2_\theta\lp \kappa h^2\rp - 2 \kappa h D^2_\theta h\rp\\
\nonumber \\
c_t &=& hn + s_0\tau.\label{eq:CHhshoot2}
\end{eqnarray}
Equations \eqref{eq:CHhshoot1} - \eqref{eq:CHhshoot2} incorporate an extra term $s_0$ compared with \eqref{eq:hshoot1} - \eqref{eq:hshoot2} which allow matching without the need to reparametrize the boundary data.  Note that nothing is asserted about the parametrization away from the boundary data.  Differences between the equations derived here and those in \cite{CoHo2009} are attributed to different approaches being used.  Their equations live on $\mX\lp \mR^2\rp^*$, the one-form densities on $\mR^2$, with the equations on $T^*Q$ being given by a momentum map.  Our equations and variational principle are formulated directly on $TQ$.  Whilst the precise relation between the two sets of equations remains unclear, the method of adding an initial vector field on $S^1$ in order to enable matching parametrized boundary data is in agreement.  Indeed, comparing equations before and after adding in the extra parameter both here and in \cite{CoHo2009} shows that each set of equations are modified by adding precisely the same terms.

\subsubsection{Uniform parametric morphing}\label{sec:UPM}

The second case considers forcing $f$ that ensures that the curve $c$ is parametrized uniformly for all time.  We call this {\em uniformly parametrized morphing}. That is, that the equation
\begin{equation}
|c_\theta| = l(c) := \int^1_0\operatorname{vol}(c)
\label{eq:constParam}
\end{equation}
holds for all time and for all $\theta \in S^1$.  Notice that the right hand side of \eqref{eq:constParam} is independent of $\theta$.  \cite{MiMu2006} shows that such curves, which are equally well described by the relation $|c_\theta| = constant$, form a submanifold $\cU\subset \Emb(S^1,\mR^2)$.  Differentiating \eqref{eq:constParam} shows that a tangent vector $(c,U) \in T\cU$ if and only if
\begin{equation}
|c_\theta| \lp D_\theta \lp U\cdot \tau\rp - \lp U\cdot n\rp\kappa\rp = -\int^1_0 \lp U \cdot n\rp \kappa \, \operatorname{vol}(c).
\label{eq:constParam1}
\end{equation}
Indeed, if $\delta c=U$, the variation of the left hand side of \eqref{eq:constParam} is
\begin{eqnarray*}
\delta |c_\theta| &=&\frac{c_\theta\cdot\delta c_\theta}{|c_\theta|}=c_\theta\cdot D_\theta \delta c=c_\theta\cdot (\delta \tau-[\delta, D_\theta]c)=|c_\theta|\lp D_\theta \lp U\cdot \tau\rp - \lp U\cdot n\rp\kappa\rp,
\end{eqnarray*}
where we used that $c_\theta\cdot\delta\tau=0$. The variation of the right hand side is
\[
\delta \int^1_0\operatorname{vol}(c)=\int_0^1\lp D_\theta\lp U\cdot \tau \rp - \lp U \cdot n\rp\kappa\rp \, \operatorname{vol}(c)=-\int_0^1\lp U \cdot n\rp \kappa\, \operatorname{vol}(c)
\]
Applying \eqref{eq:constParam1} to $U=c_t$ reveals the following relation
\begin{equation}
|c_\theta| \lp D_\theta s - h\kappa\rp = -\int^1_0 h \kappa \, \operatorname{vol}(c).
\label{eq:constParam1a}
\end{equation}
Solving the relation \eqref{eq:constParam1a} for $s$ in Proposition \ref{prop:unvar} would yield the required vertical term to keep the evolution of $c$ uniformly parametrized.  There is still a problem to deal with, namely that \eqref{eq:constParam1} does not respect the action of $G$.  This corresponds to breaking the $G$-invariance of $f$ in the un-reduction equations, which is disallowed since $G$-invariance is critical to the derivation of the un-reduction equations.  

A solution to this issue may be found by noting that one only requires \eqref{eq:constParam1} {\em on the submanifold of uniformly parametrized curves}, $\cU$.  The intersection of each fibre of $Q \to \Sigma$ with $\cU$ are the orbits of rigid rotation of $S^1$.  Indeed, \eqref{eq:constParam1} is invariant under the action of $S^1$ on $Q$ given by $(\al,c(\theta)) \mapsto c(\theta + \al)$.  Therefore, one may define $f$ on $\cU$ such that \eqref{eq:constParam1} is satisfied, and then extend to the whole of $Q$ by enforcing $G$-invariance.

Having described the way in which a suitable $f$ may be found in principle, it now appears that there is a short cut that allows the explicit calculation of $f$ to be bypassed.  Note that on the submanifold $\cU$ one has $l/|c_\theta| = 1$.  Thus one may multiply this factor on the {\em left hand side only} of \eqref{eq:constParam1} to obtain a relation that is precisely equivalent to the original on $\cU$, but is also $G$-invariant.  This relation, applied to $U=c_t$, is
\begin{equation}
l \lp D_\theta s - h\kappa\rp = -\int^1_0 h \kappa \, \operatorname{vol}(c).
\label{eq:constParam2}
\end{equation}
Solving \eqref{eq:constParam2} for $s$ yields the same result as integrating the second un-reduction equation from Proposition \ref{prop:unvar} with $f$ constructed as described above.  Rearranging \eqref{eq:constParam2} yields
\begin{equation}
D_\theta s = h\kappa - \frac{1}{l}\int^1_0 h\kappa \,\operatorname{vol}(c).
\label{eq:constParam3}
\end{equation}
Equation \eqref{eq:constParam3} states that the rate of change of $s$ as one moves along the curve is equal to the deviation of the quantity $h\kappa$ from its average around the curve.  Since uniform parametrization is a non-local property of the curve, we expect that non-local terms appear in the equations, these terms take the form of the average of $h\kappa$.

Combining \eqref{eq:constParam3} with the un-reduction equations from Proposition \ref{prop:unvar} yields the following equations
\begin{eqnarray}\label{eq:lift1}
h_t &=& D_\theta \lp s h\rp - \frac{\kappa\lp 1+ 3\kappa^2\rp}{2\lp 1 + \kappa^2\rp}h^2 + \frac{1}{\lp 1+\kappa^2\rp}\lp D^2_\theta\lp \kappa h^2\rp - 2 \kappa h D^2_\theta h\rp\\
\nonumber\\
s(\theta) &=& s(0) + \int^\theta_0h\kappa \, \operatorname{vol}(c) - \frac{1}{l} \int^\theta_0 \, \operatorname{vol}(c)\int^1_0 h \kappa \, \operatorname{vol}(c),\label{eq:lift2}\\
\nonumber\\
c_t &=& hn + s\tau.\label{eq:lift3}
\end{eqnarray}

Equations \eqref{eq:lift1}---\eqref{eq:lift3} may be described as the lifting of the geodesic equations on $\Sigma$ under the metric induced by $\cG$ to the the submanifold, $\cU \subset \Emb\lp S^1, \mR^2\rp$, of uniformly parametrized embeddings.  As far as we are aware, \eqref{eq:lift1}---\eqref{eq:lift3} have not appeared in the literature.  The main interest in these un-reduction equations is that they solve the major issue concerning the horizontal shooting method, \eqref{eq:hshoot1}---\eqref{eq:hshoot2}.  That is, the parametrization of solutions of \eqref{eq:lift1}---\eqref{eq:lift3} behave regularly for all time.  This property solves the problems concerning clustering of data points and bad parametrizations implicitly along the entire trajectory, thereby solving the parametrization issues described in \S\ref{sec:Intro}.

\subsubsection{Section morphing}

During a recent visit to London, M. Bauer and P. Harms asked whether un-reduction could be used to achieve geodesic morphing of shapes whilst respecting the property of being described as a graph.  We worked with Bauer and Harms to show that this objective is indeed possible.  The argument is presented here for the case of simple, closed planar curves.  We would like to thank Bauer and Harms for their insight and collaboration on this point, and hope that the experience is helpful for their studies of higher dimensional matching problems, see some examples of which may be found in \cite{BaHaMi2010}.

To capture the setup geometrically, we begin with slight generalizations.  Consider $Q = \Emb\lp \cM, \cN\rp$, which is a principal bundle over $\Sigma = Q/G$ where $G = \Diff\lp \cM\rp$, as in \cite{KrMi1997}.  Suppose that $\cN = \cM \times \cS$ is a trivial fibre bundle over $\cM$ with fibre $\cS$ and projection $\pi_\cN$ given by projection on the first factor.  A {\em graph} is then a section $\eta \in \Gamma(\pi_\cN)$.  We say a curve $c(t) \subset Q$ respects the graph property if and only if $c(t) \subset \Gamma(\pi_\cN) \cap Q$.  That is, $\pi_\cN\circ c(t) = \id_\cM$.

Interesting choices could involve $\lp \cM,\cN\rp = \lp S^2, \mR^3/\{0\} \equiv S^2 \times \mR^+\rp$, which may be used to match spheres embedded in $\mR^3$ under small deformations.  The example $\lp \cM, \cN\rp = \lp \mR^2, \{(x,y,z) \in \mR^3| z>0\} \equiv \mR^2 \times \mR^+\rp$ could be used to compare different topographies.

Here, we shall specialize to the case of closed, simple planar curves.  This corresponds to the particular choice $\lp \cM, \cN\rp = \lp S^1, \mR^2/\{0\}\equiv S^1 \times \mR^+\rp$.  We shall identify $S^1$ with the unit circle in $\mR^2$, and fix
\[
\pi_\cN:\cN = \mR^2/\{0\}\rightarrow \cM=S^1, \quad \pi_\cN({\bf x}) = {\bf x}/|{\bf x}|.
\]
Section morphing is not available for all curves in the plane, but rather is constrained to curves $c \in \Gamma(\pi_{\cN}) \cap Q$ that contain the origin.  Considering curves up to translations in the plane, the choice of origin becomes arbitrary, and we shall therefore ignore the origin condition.

The approach to deriving the un-reduction equations is similar to that taken for uniform parametrisation morphing in \S\ref{sec:UPM}.  That is, we derive a relation on the submanifold $\cU := \Gamma(\pi_{\mathcal{N}})\cap Q \subset Q$.  The submanifold $\cU$ is not $G$-invariant, in fact, it is only preserved by the identity.  Despite this, in contrast with \S\ref{sec:UPM}, our relation turns out to be $G$-invariant without the need for modification. Therefore, the implied forcing may be extended to the $G$-invariant submanifold $\cU\cdot G \subset Q$.  This observation again allows us to identify our equations as un-reduction equations without explicitly deriving the vertical forcing.
   

For the graph condition to remain invariant over time, the map $\pi_\cN\circ c(t)$ must be time-invariant. This leads to the following calculation
\begin{equation}\label{eq:GraphPres}
0 = \prt_t \lp \frac{c}{|c|} \rp= \frac{1}{|c|}\lp c_t - c \frac{\lp c\cdot c_t\rp}{|c|^2} \rp.
\end{equation}
From this relation we deduce that
\[
s=c_t\cdot\tau=(c\cdot\tau)\frac{c\cdot c_t}{|c|^2}\quad\text{and}\quad h=c_t\cdot n=(c\cdot n)\frac{c\cdot c_t}{|c|^2}
\]
which results in the relation
\begin{equation}
s = h\frac{(c\cdot \tau)}{(c \cdot n)}.
\label{eq:GraphTang}
\end{equation}

Note that equation \eqref{eq:GraphTang} is $G$-invariant, therefore it may be extended along the whole of $\cU \cdot G$.  Differentiating \eqref{eq:GraphTang} and substituting in equations from Proposition \ref{prop:unvar} would again result in a complicated expression for the exogenous parameter, $f$, that would be the graph preservation forcing term.  Given that \eqref{eq:GraphTang}, we omit the calculation.  Collecting the un-reduction equations from Proposition \ref{prop:unvar} together with \eqref{eq:GraphTang} yields the section morphing un-reduction equations
\begin{eqnarray}
h_t &=& D_\theta \lp h^2\frac{\lp c\cdot \tau\rp}{\lp c\cdot n\rp}\rp - \frac{\kappa\lp 1+ 3\kappa^2\rp}{2\lp 1 + \kappa^2\rp}h^2 + \frac{1}{\lp 1+\kappa^2\rp}\lp D^2_\theta\lp \kappa h^2\rp - 2 \kappa h D^2_\theta h\rp\label{eq:GraphUnred1}\\
\nonumber\\
c_t &=& \frac{h}{\lp c\cdot n\rp}c.
\label{eq:GraphUnred2}
\end{eqnarray}

Introducing polar coordinates in the plane, $c(t,\theta) = r(t,\theta)(\cos \theta, \sin \theta)$, equations \eqref{eq:GraphUnred1} - \eqref{eq:GraphUnred2} become
\begin{eqnarray}
h_t &=& D_\theta \lp h^2\frac{r_\theta}{r}\rp - \frac{\kappa\lp 1+ 3\kappa^2\rp}{2\lp 1 + \kappa^2\rp}h^2 + \frac{1}{\lp 1+\kappa^2\rp}\lp D^2_\theta\lp \kappa h^2\rp - 2 \kappa h D^2_\theta h\rp\label{eq:RGraphUnred1}\\
\nonumber\\
r_t &=& h\sqrt{\lp \frac{r_\theta}{r}\rp^2 + 1}.
\label{eq:RGraphUnred2}
\end{eqnarray}

Note that equations \eqref{eq:RGraphUnred1} - \eqref{eq:RGraphUnred2} describe a set of differential equations on functions taking values in $\mR$, whereas, the full un-reduction equations from Proposition \ref{prop:unvar} constitute a set of differential equations on functions taking values in $\mR^2$.  Therefore the particular choice of morphing that respects the section relation \eqref{eq:GraphPres} has allowed us to integrate the un-reduction equations twice.  In higher dimensional problems, this method of reduction which we call {\em section morphing} could prove invaluable for simplifying numerical implementation.  Note that, in this example, clustering of data points is prevented by the graph preservation property.  The only parametrization effect that may cause a problem is the distance between data points becomes large as the radius $\max_\theta r(\theta,t)$ increases.  This issue may be averted by first scaling the shapes to be matched, and then using section morphing.  A detailed discussion of such developments is beyond the scope of this paper, and is left to future work.

\section{Acknowledgements} 
We thank C. J. Cotter for encouraging and helpful discussions of this work as it was being developed. 
The work of FGB was partially supported by a Swiss NSF postdoctoral fellowship. 
The work of DDH was partially supported by the Royal Society of London Wolfson scheme and the European Research Council Advanced scheme.

\end{document}